\documentclass[a4paper,12pt]{article}

\usepackage{color}
\usepackage{amssymb, amsmath}

\usepackage{tikz}

\newcommand\EFFACE[1]{}

\newcommand\ESOK[1]{#1}

\newcommand\OR[1]{\overrightarrow{#1}}

\newtheorem{theorem}{Theorem}
\newtheorem{proposition}[theorem]{Proposition}

\newenvironment{proof}{
\par
\noindent {\bf Proof.}\rm}{\mbox{}\hfill$\square$\par\vskip 3mm}

\newcommand\pcn{\chi_{\rho}}

\newcommand\SOMMET[1]{\draw[fill=black] (#1) circle (2pt)}
\newcommand\ETIQUETTE[2]{\node[below] at (#1) {#2}}
\newcommand\ARETEH[2]{\draw[very thick] (#1) -- ++(#2,0)}
\newcommand\ARETEV[1]{\draw[very thick] (#1) -- ++(0,1)}
\newcommand\ARCN[1]{\draw[very thick,->] (#1) -- ++(0,0.6)}
\newcommand\ARCS[1]{\draw[very thick,->] (#1) -- ++(0,-0.6)}
\newcommand\ARCO[1]{\draw[very thick,->] (#1) -- ++(0.6,0)}
\newcommand\ARCE[1]{\draw[very thick,->] (#1) -- ++(-0.6,0)}

\def\LL{\ \longleftarrow\ }
\def\RR{\ \longrightarrow\ }

\makeatletter
\let\@fnsymbol\@arabic
\makeatother

\begin{document}


\title{\ESOK{Packing coloring of some undirected and oriented coronae graphs}}

\author{Daouya LA\"{I}CHE~\thanks{Faculty of Mathematics, Laboratory L'IFORCE, University of Sciences and Technology
Houari Boumediene (USTHB), B.P.~32 El-Alia, Bab-Ezzouar, 16111 Algiers, Algeria.}
\and Isma BOUCHEMAKH~\footnotemark[1]
\and \'Eric SOPENA~\thanks{Univ. Bordeaux, LaBRI, UMR5800, F-33400 Talence, France.}~$^,$\thanks{CNRS, LaBRI, UMR5800, F-33400 Talence, France.}~$^,$\footnote{Corresponding author. Eric.Sopena@labri.fr}
}


%

\maketitle

\abstract{
The packing chromatic number $\pcn(G)$ of a graph $G$ is the smallest integer $k$ such that
its set of vertices $V(G)$ can be partitioned into $k$ disjoint subsets $V_1$, \ldots, $V_k$, in such a way that every two distinct vertices
in $V_i$ are at distance greater than $i$ in $G$ for every $i$, $1\le i\le k$.
\ESOK{For a given integer $p\ge 1$, the generalized corona $G\odot pK_1$ of a graph $G$ is the graph obtained from $G$ by adding $p$ degree-one neighbors to every vertex  of $G$}.
\ESOK{In this paper, we} determine the packing chromatic number of \ESOK{generalized coronae} of paths and cycles.

\ESOK{Moreover, by considering digraphs 
and the (weak) directed distance
between vertices, we get a natural extension of the notion of packing coloring to digraphs.
We then determine the packing chromatic number of orientations of generalized coronae of paths and cycles.}
}

\medskip



\section{Introduction}

All the graphs we considered are \ESOK{simple and loopless}.
\ESOK{For an undirected graph $G$}, we denote by $V(G)$ its set of vertices and by $E(G)$ its set of edges.
The {\em distance} $d_G(u,v)$, or simply $d(u,v)$, between
vertices $u$ and $v$ in $G$ is the length \ESOK{(number of edges)} of a shortest path joining $u$ and $v$.
The {\em diameter} 
of $G$ is the maximum distance between two vertices of $G$.
We denote by $P_n$ the path of order $n$ and by $C_n$, $n\ge 3$, the cycle of order $n$.

A {\em packing $k$-coloring} of $G$ is a mapping $\pi:V(G)\rightarrow\{1,\ldots,k\}$
such that, for every two distinct vertices $u$ and $v$,  \ESOK{$\pi(u)=\pi(v)=i$ implies $d(u,v)>i$}.
The {\em packing chromatic number} $\pcn(G)$ of $G$ is then the smallest $k$ such that
$G$ admits a packing $k$-coloring.
In other words, $\pcn(G)$ is the smallest integer $k$ such that
$V(G)$ can be partitioned into $k$ disjoint subsets $V_1$, \ldots, $V_k$, in such a way that every two vertices
in $V_i$ are at distance greater than $i$ in $G$ for every $i$, $1\le i\le k$.
\ESOK{A packing coloring of $G$ is {\em optimal} if it uses exactly $\pcn(G)$ colors.}

Packing coloring has been introduced by Goddard, Hedetniemi, Hedetniemi, Harris and Rall~\cite{GHHHR03,GHHHR08} under the name {\em broadcast coloring}
and has been studied by several authors in recent years.
Several papers deal with the packing chromatic number of certain classes of graphs such
as trees~\cite{ANT14,BKR07,GHHHR08,RBFK08,S04},
lattices~\cite{BKR07,EFHL10,FKL09,FR10,KV14,SH10},
Cartesian products~\cite{BKR07,FKL09,RBFK08},
distance graphs~\cite{EHL12,EHT14,T14} or hypercubes~\cite{GHHHR08,TVP15,WRR14b}.
Complexity issues of the packing coloring problem were adressed
in~\cite{A11,ANT12,ANT14,FG10,G15,GHHHR08}.

The following proposition, which states that having packing chromatic number at most $k$
is a hereditary property, will be useful in the sequel:

\begin{proposition}[Goddard, Hedetniemi, Hedetniemi, Harris and Rall~\cite{GHHHR08}]
If $H$ is a subgraph of $G$, then $\pcn(H)\le\pcn(G)$.
\label{prop:subgraph}
\end{proposition}

Fiala and Golovach~\cite{FG10} proved that determining the packing chromatic number is an NP-hard problem for trees.
Determining the packing chromatic number of special subclasses of trees is thus an interesting problem.
The exact value of the packing chromatic number of trees with diameter at most 4 was given in~\cite{GHHHR08}.
In the same paper, it was proved that $\pcn(T_n)\le(n+7)/4$ for every tree $T_n$ or order $n\neq 4,8$, and this bound is tight,
while $\pcn(T_n)\le 3$ if $n=4$ and $\pcn(T_n)\le 4$ if $n=8$, these two bounds being also tight.

The packing chromatic numbers of 
paths and cycles have been
determined by Goddard {\em et al.}:

\begin{theorem}[Goddard, Hedetniemi, Hedetniemi, Harris and Rall~\cite{GHHHR08}]\mbox{}
\begin{itemize}
\item $\pcn(P_n)=2$ if $n\in\{2,3\}$, 
\item $\pcn(P_n)=3$ if \ESOK{$n\ge 4$},
\item $\pcn(C_n)=3$ if $n=3$ or $n\equiv 0\pmod 4$,
\item $\pcn(C_n)=4$ if $n\ge 5$ and $n\equiv 1,2,3\pmod 4$. 
\end{itemize}
\label{th:goddard}
\end{theorem}

The {\em corona} $G\odot K_1$ of a graph $G$ is the graph obtained from $G$ by adding a degree-one neighbor to every vertex  of $G$.
We call such a degree-one neighbor a {\em pendant vertex} or a {\em pendant neighbor}.
\ESOK{More generally, for a given integer $p\ge 1$, the {\em generalized corona} $G\odot pK_1$ of a graph $G$ is the graph obtained from $G$ by adding $p$ pendant neighbors
 to every vertex  of $G$.
}

A {\em caterpillar of length} $\ell\ge 1$ is a tree whose set of internal vertices
(vertices with degree at least 2) induces a path of length $\ell-1$, 
called the {\em central path}.
Sloper proved the following result:

\begin{theorem}[Sloper~\cite{S04}]
Let $CT_\ell$ be a caterpillar of length $\ell$. Then $\pcn(CT_\ell)\le 6$ if $\ell\le 34$,
and $\pcn(CT_\ell)\le 7$ otherwise. Moreover, these two bounds are tight.
\label{th:sloper}
\end{theorem}

\ESOK{Since every generalized corona of a path is a caterpillar, we get that 
for every integer $p\ge 1$,  
$\pcn(P_n\odot pK_1)\le 6$ if $n\le 34$ and $\pcn(P_n\odot pK_1)\le 7$ otherwise.}


\medskip

\ESOK{By considering digraphs instead of undirected graphs, and using the {\em (weak) directed distance} between vertices
--- defined as the number of arcs in a shortest directed path linking these vertices, in either direction ---
we get a natural extension of packing colorings to digraphs. In this paper, we will consider {\em orientations}
of some undirected graphs, obtained by giving to each edge of such a graph one of its two possible orientations.
The so-obtained {\em oriented graphs} are thus digraphs having no pair of opposite arcs.
}

\medskip

\ESOK{In this paper, we determine the packing chromatic number of (simple) coronae of paths and cycles (Section~\ref{sec:coronae}) 
and of generalized coronae (for $k\ge 2$) of paths and cycles (Section~\ref{sec:generalized-coronae}).
In Section~\ref{sec:oriented-coronae}, we consider the oriented version of packing colorings and
determine the packing chromatic number of oriented paths, oriented cycles and oriented generalized coronae
of paths and cycles.
}%
Some of the presented results \ESOK{for undirected graphs} were obtained by the first author in~\cite{L10}.
\EFFACE{and characterize caterpillars with packing chromatic number at most 4, 5 and~6 (Section~\ref{sec:caterpillars}). }

\section{Coronae of \ESOK{undirected} paths and cycles}
\label{sec:coronae}

%

We study in this section coronae of paths and cycles.
We first determine the packing chromatic number of coronae of paths.
Note that any corona $P_n\odot K_1$ is also a caterpillar of length $n$.

\begin{theorem}
The packing chromatic number of the corona graph $P_n\odot K_1$ is given by:
$$\pcn(P_{n}\odot K_1)= \left\{
			    \begin{array}{ll}
    			     $2$ & \hbox{if $n=1$,} \\
    			     $3$ & \hbox{if $n\in\{2,3\}$,} \\
    			     $4$ & \hbox{if $4\le n\le 9$,}\\
    			     $5$ & \hbox{if $n\ge 10$.}
   			 	\end{array}
 			 			\right.
$$ 
\label{th:CrPn}
\end{theorem}

\begin{figure}
\begin{center}
\begin{tikzpicture}[domain=0:17,x=0.8cm,y=0.8cm]
\begin{scope}
\draw[fill=black] (0,0) circle (2pt); \draw[fill=black] (1,0) circle
(2pt); \draw[fill=black] (3,0) circle (2pt); \draw[fill=black] (4,0) circle
(2pt); \draw[fill=black] (5,0) circle (2pt);
\draw[fill=black] (7,0) circle (2pt);\draw[fill=black] (8,0) circle (2pt);\draw[fill=black] (9,0) circle (2pt); \draw[fill=black] (10,0) circle (2pt);
\draw[fill=black] (12,0) circle (2pt);\draw[fill=black] (13,0) circle (2pt); \draw[fill=black] (14,0) circle (2pt);\draw[fill=black] (15,0) circle (2pt);
\draw[fill=black] (16,0) circle (2pt);

\draw[fill=black] (0,1) circle (2pt); \draw[fill=black] (1,1) circle
(2pt); \draw[fill=black] (3,1) circle (2pt); \draw[fill=black] (4,1) circle
(2pt); \draw[fill=black] (5,1) circle (2pt); \draw[fill=black] (7,1) circle
(2pt); \draw[fill=black] (8,1) circle (2pt);\draw[fill=black] (9,01) circle (2pt); \draw[fill=black] (10,1) circle
 (2pt);
\draw[fill=black] (12,1) circle (2pt);\draw[fill=black] (13,1) circle (2pt);\draw[fill=black] (14,1) circle (2pt);
\draw[fill=black] (15,1) circle (2pt);
\draw[fill=black] (16,1) circle (2pt);

\draw[fill=black] (2,-2.1) circle (2pt);\draw[fill=black] (3,-2.1) circle (2pt);\draw[fill=black] (4,-2.1) circle (2pt);\draw[fill=black] (5,-2.1) circle (2pt);\draw[fill=black] (6,-2.1) circle (2pt);\draw[fill=black] (7,-2.1) circle (2pt);
\draw[fill=black] (9,-2.1) circle (2pt);\draw[fill=black] (10,-2.1) circle (2pt);\draw[fill=black] (11,-2.1) circle (2pt);\draw[fill=black] (12,-2.1) circle (2pt);\draw[fill=black] (13,-2.1) circle (2pt);\draw[fill=black] (14,-2.1) circle (2pt);\draw[fill=black] (15,-2.1) circle (2pt);

\draw[fill=black] (2,-3.1) circle (2pt);\draw[fill=black] (3,-3.1) circle (2pt);\draw[fill=black] (4,-3.1) circle (2pt);\draw[fill=black] (5,-3.1) circle (2pt);\draw[fill=black] (6,-3.1) circle (2pt);\draw[fill=black] (7,-3.1) circle (2pt);
\draw[fill=black] (9,-3.1) circle (2pt);\draw[fill=black] (10,-3.1) circle (2pt);\draw[fill=black] (11,-3.1) circle (2pt);\draw[fill=black] (12,-3.1) circle (2pt);\draw[fill=black] (13,-3.1) circle (2pt);\draw[fill=black] (14,-3.1) circle (2pt);\draw[fill=black] (15,-3.1) circle (2pt);

\node[below] at (0,1.7) {$2$};\node[below] at (0,-0.1) {$1$};
\node[below] at (1,1.7) {$1$};\node[below] at (1,-0.1) {$3$};

\node[below] at (3,1.7) {$1$};\node[below] at (3,-0.1) {$2$};
\node[below] at (4,1.7) {$3$};\node[below] at (4,-0.1) {$1$};
\node[below] at (5,1.7) {$1$};\node[below] at (5,-0.1) {$2$};

\node[below] at (7,1.7) {$2$};\node[below] at (7,-0.1) {$1$};
\node[below] at (8,1.7) {$1$};\node[below] at (8,-0.1) {$3$};
\node[below] at (9,1.7) {$4$};\node[below] at (9,-0.1) {$1$};
\node[below] at (10,1.7) {$1$};\node[below] at (10,-0.1) {$2$};
%
%

\node[below] at (0.5,-0.7) {$P_{2}\odot K_1$};
\node[below] at (4,-0.7) {$P_{3}\odot K_1$};
\node[below] at (08.5,-0.7) {$P_{4}\odot K_1$};
\node[below] at (14,-0.7) {$P_{5}\odot K_1$};

\node[below] at (4.5,-3.8) {$P_{6}\odot K_1$};
\node[below] at (12,-3.8) {$P_{7}\odot K_1$};

\node[below] at (12,1.7) {$2$};\node[below] at (12,-0.1) {$1$};
\node[below] at (13,1.7) {$1$};\node[below] at (13,-0.1) {$3$};
\node[below] at (14,1.7) {$4$};\node[below] at (14,-0.1) {$1$};
\node[below] at (15,1.7) {$1$};\node[below] at (15,-0.1) {$3$};
\node[below] at (16,1.7) {$2$};\node[below] at (16,-0.1) {$1$};

\node[below] at (2,-1.4) {$2$};\node[below] at (2,-3.2) {$1$};
\node[below] at (3,-1.4) {$1$};\node[below] at (3,-3.2) {$3$};
\node[below] at (4,-1.4) {$4$};\node[below] at (4,-3.2) {$1$};
\node[below] at (5,-1.4) {$1$};\node[below] at (5,-3.2) {$3$};
\node[below] at (6,-1.4) {$2$};\node[below] at (6,-3.2) {$1$};
\node[below] at (7,-1.4) {$1$};\node[below] at (7,-3.2) {$3$};

\node[below] at (9,-1.4) {$1$};\node[below] at (9,-3.2) {$3$};
\node[below] at (10,-1.4) {$2$};\node[below] at (10,-3.2) {$1$};
\node[below] at (11,-1.4) {$1$};\node[below] at (11,-3.2) {$3$};
\node[below] at (12,-1.4) {$4$};\node[below] at (12,-3.2) {$1$};
\node[below] at (13,-1.4) {$1$};\node[below] at (13,-3.2) {$3$};
\node[below] at (14,-1.4) {$2$};\node[below] at (14,-3.2) {$1$};
\node[below] at (15,-1.4) {$1$};\node[below] at (15,-3.2) {$3$};
%
%
%
\draw[very thick] (0,1) -- (1,1); \draw[very thick] (3,1) -- (5,1);\draw[very thick] (7,1) -- (10,1); \draw[very thick] (12,1) -- (16,1);

\draw[very thick] (2,-2.1) -- (7,-2.1); \draw[very thick] (9,-2.1) -- (15,-2.1);
\draw[very thick] (0,0) -- (0,1); \draw[very thick] (1,0) -- (1,1); \draw[very thick] (3,0) -- (3,1);
\draw[very thick] (4,0) -- (4,1); \draw[very thick] (5,0) -- (5,1);
\draw[very thick] (7,0) -- (7,1); \draw[very thick] (8,0) -- (8,1); \draw[very thick] (9,0) -- (9,1);
 \draw[very thick] (10,0) -- (10,1);
 \draw[very thick] (12,0) -- (12,1);\draw[very thick] (13,0) -- (13,1);\draw[very thick] (14,0) -- (14,1);
 \draw[very thick] (15,0) -- (15,1);\draw[very thick] (16,0) -- (16,1);

\draw[very thick] (2,-2.1) -- (2,-3.1); \draw[very thick] (3,-2.1) -- (3,-3.1);
\draw[very thick] (4,-2.1) -- (4,-3.1); \draw[very thick] (5,-2.1) -- (5,-3.1);\draw[very thick] (6,-2.1) -- (6,-3.1);
\draw[very thick] (7,-2.1) -- (7,-3.1);
\draw[very thick] (9,-2.1) -- (9,-3.1);
 \draw[very thick] (10,-2.1) -- (10,-3.1);\draw[very thick] (11,-2.1) -- (11,-3.1);
 \draw[very thick] (12,-2.1) -- (12,-3.1);\draw[very thick] (13,-2.1) -- (13,-3.1);
 \draw[very thick] (14,-2.1) -- (14,-3.1);\draw[very thick] (15,-2.1) -- (15,-3.1);
%


\end{scope}

\begin{scope}[shift={(0,-6.2)}]
\SOMMET{0,0}; \SOMMET{0,1};
\SOMMET{1,0}; \SOMMET{1,1};
\SOMMET{2,0}; \SOMMET{2,1};
\SOMMET{3,0}; \SOMMET{3,1};
\SOMMET{4,0}; \SOMMET{4,1};
\SOMMET{5,0}; \SOMMET{5,1};
\SOMMET{6,0}; \SOMMET{6,1};
\SOMMET{7,0}; \SOMMET{7,1};

\ETIQUETTE{0,1.7}{1}; \ETIQUETTE{0,-0.1}{2};
\ETIQUETTE{1,1.7}{3}; \ETIQUETTE{1,-0.1}{1};
\ETIQUETTE{2,1.7}{1}; \ETIQUETTE{2,-0.1}{2};
\ETIQUETTE{3,1.7}{4}; \ETIQUETTE{3,-0.1}{1};
\ETIQUETTE{4,1.7}{1}; \ETIQUETTE{4,-0.1}{2};
\ETIQUETTE{5,1.7}{3}; \ETIQUETTE{5,-0.1}{1};
\ETIQUETTE{6,1.7}{2}; \ETIQUETTE{6,-0.1}{1};
\ETIQUETTE{7,1.7}{1}; \ETIQUETTE{7,-0.1}{4};

\ARETEH{0,1}{7};
\ARETEV{0,0};
\ARETEV{1,0};
\ARETEV{2,0};
\ARETEV{3,0};
\ARETEV{4,0};
\ARETEV{5,0};
\ARETEV{6,0};
\ARETEV{7,0};

\ETIQUETTE{3.5,-0.7}{$P_{8}\odot K_1$};

\end{scope}

\begin{scope}[shift={(9,-6.2)}]
\SOMMET{0,0}; \SOMMET{0,1};
\SOMMET{1,0}; \SOMMET{1,1};
\SOMMET{2,0}; \SOMMET{2,1};
\SOMMET{3,0}; \SOMMET{3,1};
\SOMMET{4,0}; \SOMMET{4,1};
\SOMMET{5,0}; \SOMMET{5,1};
\SOMMET{6,0}; \SOMMET{6,1};
\SOMMET{7,0}; \SOMMET{7,1};
\SOMMET{8,0}; \SOMMET{8,1};

\ETIQUETTE{0,1.7}{1}; \ETIQUETTE{0,-0.1}{4};
\ETIQUETTE{1,1.7}{2}; \ETIQUETTE{1,-0.1}{1};
\ETIQUETTE{2,1.7}{3}; \ETIQUETTE{2,-0.1}{1};
\ETIQUETTE{3,1.7}{1}; \ETIQUETTE{3,-0.1}{2};
\ETIQUETTE{4,1.7}{4}; \ETIQUETTE{4,-0.1}{1};
\ETIQUETTE{5,1.7}{1}; \ETIQUETTE{5,-0.1}{2};
\ETIQUETTE{6,1.7}{3}; \ETIQUETTE{6,-0.1}{1};
\ETIQUETTE{7,1.7}{2}; \ETIQUETTE{7,-0.1}{1};
\ETIQUETTE{8,1.7}{1}; \ETIQUETTE{8,-0.1}{4};

\ARETEH{0,1}{8};
\ARETEV{0,0};
\ARETEV{1,0};
\ARETEV{2,0};
\ARETEV{3,0};
\ARETEV{4,0};
\ARETEV{5,0};
\ARETEV{6,0};
\ARETEV{7,0};
\ARETEV{8,0};

\ETIQUETTE{4,-0.7}{$P_{9}\odot K_1$};

\end{scope}

\end{tikzpicture}
\caption{Optimal packing colorings of $P_n\odot K_1$, $2\le n\le 9$}
\label{fig:CrPnInf9}
\end{center}
\end{figure}
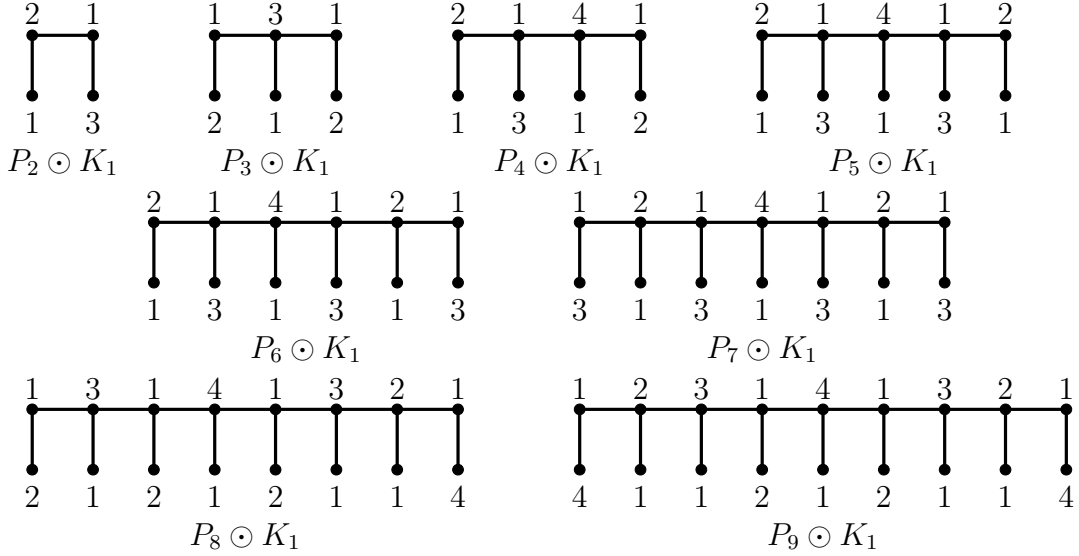

\begin{proof}
We obviously have $\pcn(P_1\odot K_1)=\pcn(P_2)=2$.
Optimal packing colorings of $P_n\odot K_1$ are given in Figure~\ref{fig:CrPnInf9} for every $n$, $2\le n\le 9$.
Since $P_2\odot K_1=P_4$, we have $\pcn(P_2\odot K_1)=3$ by Theorem~\ref{th:goddard}.
It is easy to observe that the packing 3-coloring of $P_3\odot K_1$ depicted in Figure~\ref{fig:CrPnInf9}
is unique. Hence, if $P_4\odot K_1$ would be packing 3-colorable, this packing 3-coloring of $P_3\odot K_1$
would appear on the left or right hand side of $P_4\odot K_1$. But in that case, the fourth vertex of
the central path of $P_4\odot K_1$ could not be colored. Hence  $\pcn(P_4\odot K_1)=4$.
Finally, since $P_2\odot K_1$ is a subgraph of $P_3\odot K_1$ and $P_4\odot K_1$ is a
subgraph of $P_n\odot K_1$ for every $n$, $5\le n\le 9$, all the packing colorings
given in Figure~\ref{fig:CrPnInf9} are optimal by Proposition~\ref{prop:subgraph}.

Let us now consider $P_n\odot K_1$ with $n\ge 10$.
Let $x_1x_2\ldots x_n$ denote the central path of $P_n\odot K_1$ and $y_i$ denote the pendant neighbor of $x_i$
for every $i$, $1\le i\le n$.
Let $\pi$ be the 4-periodic 5-coloring of $P_n\odot K_1$ defined as follows (see Figure~\ref{fig:CrPn}):
$$\pi(x_{i})=\left\{
   \begin{array}{ll}
   1 & \mbox{if}\ \ i\equiv 1\pmod 2,\\
   2 & \mbox{if}\ \ i\equiv 2\pmod 4,\\
   3 & \mbox{if}\ \ i\equiv 0\pmod 4,\\
   \end{array}
 \right.$$
$$\pi(y_{i})=\left\{
   \begin{array}{ll}
   1 & \mbox{if}\ \ i\equiv 0\pmod 2,\\
   4 & \mbox{if}\ \ i\equiv 1\pmod 4,\\
   5 & \mbox{if}\ \ i\equiv 3\pmod 4,\\
   \end{array}
 \right.$$
It is not difficult to check that $\pi$ is indeed a packing 5-coloring of $P_n\odot K_1$
 and, therefore,
$\pcn(P_n\odot K_1)\le 5$  for every $n\ge 10$.

To finish the proof, it is enough to prove that $\pcn(P_{10}\odot K_1)\ge 5$, thanks to Proposition~\ref{prop:subgraph}.
This could be done by a long and tedious case analysis.
By computer search, we get that the largest packing 4-colorable corona of path is $P_9\odot K_1$, which admits
two distinct packing 4-colorings:  one is given in Figure~\ref{fig:CrPnInf9}, the other one is obtained
by coloring the middle pendant vertex by 2 instead of 1.
\end{proof}

\begin{figure}
\begin{center}
\begin{tikzpicture}[domain=0:16,x=0.8cm,y=0.8cm]
\draw[fill=black] (0,0) circle (2pt); \draw[fill=black] (1,0) circle
(2pt); \draw[fill=black] (2,0) circle (2pt); \draw[fill=black] (3,0) circle (2pt); \draw[fill=black] (4,0) circle
(2pt); \draw[fill=black] (5,0) circle (2pt); \draw[fill=black] (6,0) circle (2pt); \draw[fill=black] (7,0) circle
(2pt);


\draw[fill=black] (0,1) circle (2pt); \draw[fill=black] (1,1) circle
(2pt); \draw[fill=black] (2,1) circle (2pt); \draw[fill=black] (3,1) circle (2pt); \draw[fill=black] (4,1) circle
(2pt); \draw[fill=black] (5,1) circle (2pt); \draw[fill=black] (6,1) circle (2pt); \draw[fill=black] (7,1) circle
(2pt); 


\node[below] at (0,1.7) {$1$};\node[below] at (0,-0.1) {$4$};
\node[below] at (1,1.7) {$2$};\node[below] at (1,-0.1) {$1$};
\node[below] at (2,1.7) {$1$};\node[below] at (2,-0.1) {$5$};
\node[below] at (3,1.7) {$3$};\node[below] at (3,-0.1) {$1$};
\node[below] at (4,1.7) {$1$};\node[below] at (4,-0.1) {$4$};
\node[below] at (5,1.7) {$2$};\node[below] at (5,-0.1) {$1$};
\node[below] at (6,1.7) {$1$};\node[below] at (6,-0.1) {$5$};
\node[below] at (7,1.7) {$3$};\node[below] at (7,-0.1) {$1$};
\node[below] at (8,1.15) {$...$};
\node[below] at (8,0.15) {$...$};

\draw[very thick] (0,1) -- (1,1); \draw[very thick] (1,1) -- (2,1);\draw[very thick] (2,1) -- (3,1); \draw[very thick] (3,1) -- (4,1);
\draw[very thick] (4,1) -- (5,1); \draw[very thick] (5,1) -- (6,1); \draw[very thick] (6,1) -- (7,1); \draw[very thick] (7,1) -- (7.5,1); 

\draw[very thick] (0,0) -- (0,1); \draw[very thick] (1,0) -- (1,1);\draw[very thick] (2,0) -- (2,1); \draw[very thick] (3,0) -- (3,1);
\draw[very thick] (4,0) -- (4,1); \draw[very thick] (5,0) -- (5,1); \draw[very thick] (6,0) -- (6,1); \draw[very thick] (7,0) -- (7,1); 

\draw[very thick,dashed] (3.5,-0.5) -- (3.5,1.5); \draw[very thick,dashed] (7.5,-0.5) -- (7.5,1.5); 

\end{tikzpicture}
\caption{Periodic packing coloring of $P_n\odot K_1$, $n\ge 8$}
\label{fig:CrPn}
\end{center}
\end{figure}
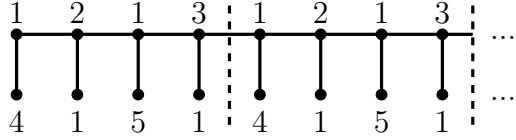

In~\cite{WRR14a}, William, Roy and Rajasingh proved that $\pcn(C_n\odot K_1)\le 5$ for every even $n\ge 6$.
We complete their result as follows:

\begin{theorem}
The packing chromatic number of the corona graph $C_n\odot K_1$ is given by:
$$\pcn(C_{n}\odot K_1)= \left\{
			    \begin{array}{ll}
    			     $4$ & \hbox{if $n\in\{3,4\}$,} \\
    			     $5$ & \hbox{if $n\ge 5$.}
   			 	\end{array}
 			 			\right.
$$ 
\label{th:CrCn}
\end{theorem}

\begin{figure}
\begin{center}
\begin{tikzpicture}[domain=0:16,x=0.8cm,y=0.8cm]
\draw[fill=black] (0,0) circle (2pt); \draw[fill=black] (1,0) circle
(2pt); \draw[fill=black] (2,0) circle (2pt);
\draw[fill=black] (04,0) circle (2pt);\draw[fill=black] (05,0) circle (2pt);\draw[fill=black] (6,0) circle (2pt); \draw[fill=black] (7,0) circle (2pt);

\draw[fill=black] (0,1) circle (2pt); \draw[fill=black] (1,1) circle
(2pt); \draw[fill=black] (2,1) circle (2pt);
\draw[fill=black] (4,1) circle (2pt);\draw[fill=black] (5,1) circle (2pt);\draw[fill=black] (6,1) circle (2pt); \draw[fill=black] (7,1) circle (2pt);

\node[below] at (0,1.75) {$2$};\node[below] at (0,-0.1) {$1$};
\node[below] at (1,1.7) {$3$};\node[below] at (1,-0.1) {$1$};
\node[below] at (2,1.75) {$4$};\node[below] at (2,-0.1) {$1$};

\node[below] at (4,1.7) {$4$};\node[below] at (4,-0.1) {$1$};
\node[below] at (5,1.7) {$1$};\node[below] at (5,-0.1) {$2$};
\node[below] at (6,1.7) {$3$};\node[below] at (6,-0.1) {$1$};
\node[below] at (7,1.7) {$1$};\node[below] at (7,-0.1) {$2$};

\node[below] at (1,-0.81) {$C_{3}\odot K_1$};\node[below] at (5.5,-0.81) {$C_{4}\odot K_1$};

\draw[very thick] (0,1) -- (2,1);\draw[very thick] (4,1) -- (7,1); 
\draw[very thick] (0,0) -- (0,1); \draw[very thick] (1,0) -- (1,1);\draw[very thick] (2,0) -- (2,1);
\draw[very thick] (4,0) -- (4,1);\draw[very thick] (5,0) -- (5,1);\draw[very thick] (6,0) -- (6,1); \draw[very thick] (7,0) -- (7,1);

\draw[very thick] (0,1) .. controls (0.5,1.8) and (1.5,1.8) .. (2,1);
\draw[very thick] (4,1) .. controls (5,2) and (6,2) .. (7,1);
\end{tikzpicture}
\caption{Optimal packing colorings of $C_3\odot K_1$ and $C_4\odot K_1$}
\label{fig:CrC3C4}
\end{center}
\end{figure}
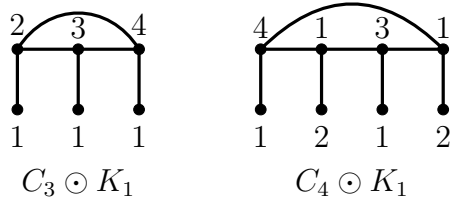

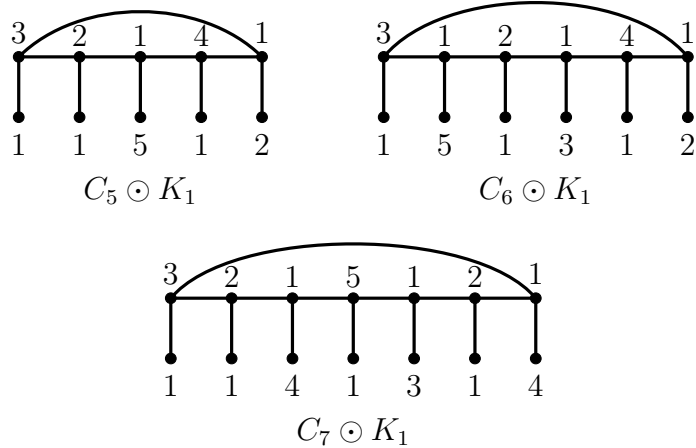
\begin{figure}
\begin{center}
\begin{tikzpicture}[domain=0:11,x=0.8cm,y=0.8cm]

\SOMMET{0,0}; \SOMMET{0,1};
\SOMMET{1,0}; \SOMMET{1,1};
\SOMMET{2,0}; \SOMMET{2,1};
\SOMMET{3,0}; \SOMMET{3,1};
\SOMMET{4,0}; \SOMMET{4,1};

\SOMMET{6,0}; \SOMMET{6,1};
\SOMMET{7,0}; \SOMMET{7,1};
\SOMMET{8,0}; \SOMMET{8,1};
\SOMMET{9,0}; \SOMMET{9,1};
\SOMMET{10,0}; \SOMMET{10,1};
\SOMMET{11,0}; \SOMMET{11,1};

\SOMMET{2.5,-4}; \SOMMET{2.5,-3};
\SOMMET{3.5,-4}; \SOMMET{3.5,-3};
\SOMMET{4.5,-4}; \SOMMET{4.5,-3};
\SOMMET{5.5,-4}; \SOMMET{5.5,-3};
\SOMMET{6.5,-4}; \SOMMET{6.5,-3};
\SOMMET{7.5,-4}; \SOMMET{7.5,-3};
\SOMMET{8.5,-4}; \SOMMET{8.5,-3};


\ETIQUETTE{0,1.75}{3}; \ETIQUETTE{0,-0.1}{1};
\ETIQUETTE{1,1.7}{2}; \ETIQUETTE{1,-0.1}{1};
\ETIQUETTE{2,1.7}{1}; \ETIQUETTE{2,-0.1}{5};
\ETIQUETTE{3,1.7}{4}; \ETIQUETTE{3,-0.1}{1};
\ETIQUETTE{4,1.75}{1}; \ETIQUETTE{4,-0.1}{2};

\ETIQUETTE{6,1.75}{3}; \ETIQUETTE{6,-0.1}{1};
\ETIQUETTE{7,1.7}{1}; \ETIQUETTE{7,-0.1}{5};
\ETIQUETTE{8,1.7}{2}; \ETIQUETTE{8,-0.1}{1};
\ETIQUETTE{9,1.7}{1}; \ETIQUETTE{9,-0.1}{3};
\ETIQUETTE{10,1.7}{4}; \ETIQUETTE{10,-0.1}{1};
\ETIQUETTE{11,1.75}{1}; \ETIQUETTE{11,-0.1}{2};

\ETIQUETTE{2.5,-2.25}{3}; \ETIQUETTE{2.5,-4.1}{1};
\ETIQUETTE{3.5,-2.3}{2}; \ETIQUETTE{3.5,-4.1}{1};
\ETIQUETTE{4.5,-2.3}{1}; \ETIQUETTE{4.5,-4.1}{4};
\ETIQUETTE{5.5,-2.3}{5}; \ETIQUETTE{5.5,-4.1}{1};
\ETIQUETTE{6.5,-2.3}{1}; \ETIQUETTE{6.5,-4.1}{3};
\ETIQUETTE{7.5,-2.3}{2}; \ETIQUETTE{7.5,-4.1}{1};
\ETIQUETTE{8.5,-2.25}{1}; \ETIQUETTE{8.5,-4.1}{4};

\ETIQUETTE{2,-0.81}{$C_{5}\odot K_1$};
\ETIQUETTE{8.5,-0.81}{$C_{6}\odot K_1$};
\ETIQUETTE{5.5,-4.81}{$C_{7}\odot K_1$};


\ARETEH{0,1}{4};
\ARETEV{0,0};
\ARETEV{1,0};
\ARETEV{2,0};
\ARETEV{3,0};
\ARETEV{4,0};

\ARETEH{6,1}{5};
\ARETEV{6,0};
\ARETEV{7,0};
\ARETEV{8,0};
\ARETEV{9,0};
\ARETEV{10,0};
\ARETEV{11,0};

\ARETEH{2.5,-3}{6};
\ARETEV{2.5,-4};
\ARETEV{3.5,-4};
\ARETEV{4.5,-4};
\ARETEV{5.5,-4};
\ARETEV{6.5,-4};
\ARETEV{7.5,-4};
\ARETEV{8.5,-4};


\draw[very thick] (0,1) .. controls (1,2) and (3,2) .. (4,1);

\draw[very thick] (6,1) .. controls (7,2.2) and (10,2.2) .. (11,1);

\draw[very thick] (2.5,-3) .. controls (3.5,-1.8) and (7.5,-1.8) .. (8.5,-3);

\end{tikzpicture}
\caption{Optimal packing colorings of $C_5\odot K_1$, $C_6\odot K_1$ and $C_7\odot K_1$}
\label{fig:CrC5C6C7}
\end{center}
\end{figure}

\begin{proof}
Optimal packing 4-colorings of $C_3\odot K_1$ and $C_4\odot K_1$ are given in Figure~\ref{fig:CrC3C4}.
We claim indeed that these two coronae graphs cannot be packing 3-colored.
If there would exist such colorings then color 1 would necessarily be used for the cycle
and its two neighbors on the cycle would get colors 2 and 3. But then, it would not be possible
to color the pendant neighbor of the vertex with color 1.

Let us now consider $C_n\odot K_1$ with $n\ge 5$.
Figure~\ref{fig:CrC5C6C7} describes 5-colorings of $C_5\odot K_1$, $C_6\odot K_1$ and $C_7\odot K_1$.
Figure~\ref{fig:CrCn} describes ``almost 4-periodic'' packing 5-colorings of $C_n\odot K_1$, $n\ge 8$,
according to the value of $n$ mod 4 (the leftmost pattern of length 4 can be repeated any number of times).
It is not difficult to check that all these colorings are indeed  packing 5-colorings 
and, therefore, $\pcn(C_n\odot K_1)\le 5$  for every $n\ge 5$.

It remains to prove that $\pcn(C_n\odot K_1)\ge 5$  for every $n\ge 5$.
Assume to the contrary that there exists a packing 4-coloring of 
$C_5\odot K_1$. By ``unfolding'' this coloring and considering it as a pattern of
a 5-periodic coloring for coronae of paths we obtain a packing 4-coloring of every corona graph $P_n\odot K_1$, $n\ge 5$,
in contradiction with Theorem~\ref{th:CrPn}. The same argument proves that there is no packing 4-coloring
of $C_n\odot K_1$ for every $n\ge 6$. This completes the proof.
\end{proof}

\begin{figure}
\begin{center}
\begin{tikzpicture}[domain=0:12,x=0.8cm,y=0.8cm]

\SOMMET{0,0}; \SOMMET{0,1};
\SOMMET{1,0}; \SOMMET{1,1};
\SOMMET{2,0}; \SOMMET{2,1};
\SOMMET{3,0}; \SOMMET{3,1};

\SOMMET{6,0}; \SOMMET{6,1};
\SOMMET{7,0}; \SOMMET{7,1};
\SOMMET{8,0}; \SOMMET{8,1};
\SOMMET{9,0}; \SOMMET{9,1};
\SOMMET{10,0}; \SOMMET{10,1};
\SOMMET{11,0}; \SOMMET{11,1};
\SOMMET{12,0}; \SOMMET{12,1};

\SOMMET{0.5,5}; \SOMMET{0.5,6};
\SOMMET{1.5,5}; \SOMMET{1.5,6};
\SOMMET{2.5,5}; \SOMMET{2.5,6};
\SOMMET{3.5,5}; \SOMMET{3.5,6};

\SOMMET{6.5,5}; \SOMMET{6.5,6};
\SOMMET{7.5,5}; \SOMMET{7.5,6};
\SOMMET{8.5,5}; \SOMMET{8.5,6};
\SOMMET{9.5,5}; \SOMMET{9.5,6};
\SOMMET{10.5,5}; \SOMMET{10.5,6};
\SOMMET{11.5,5}; \SOMMET{11.5,6};

\SOMMET{1,10}; \SOMMET{1,11};
\SOMMET{2,10}; \SOMMET{2,11};
\SOMMET{3,10}; \SOMMET{3,11};
\SOMMET{4,10}; \SOMMET{4,11};

\SOMMET{7,10}; \SOMMET{7,11};
\SOMMET{8,10}; \SOMMET{8,11};
\SOMMET{9,10}; \SOMMET{9,11};
\SOMMET{10,10}; \SOMMET{10,11};
\SOMMET{11,10}; \SOMMET{11,11};

\SOMMET{1.5,15}; \SOMMET{1.5,16};
\SOMMET{2.5,15}; \SOMMET{2.5,16};
\SOMMET{3.5,15}; \SOMMET{3.5,16};
\SOMMET{4.5,15}; \SOMMET{4.5,16};

\SOMMET{7.5,15}; \SOMMET{7.5,16};
\SOMMET{8.5,15}; \SOMMET{8.5,16};
\SOMMET{9.5,15}; \SOMMET{9.5,16};
\SOMMET{10.5,15}; \SOMMET{10.5,16};


\ETIQUETTE{0,1.75}{3}; \ETIQUETTE{0,-0.1}{1};
\ETIQUETTE{1,1.7}{1}; \ETIQUETTE{1,-0.1}{5};
\ETIQUETTE{2,1.7}{2}; \ETIQUETTE{2,-0.1}{1};
\ETIQUETTE{3,1.7}{1}; \ETIQUETTE{3,-0.1}{4};

\ETIQUETTE{6,1.7}{3}; \ETIQUETTE{6,-0.1}{1};
\ETIQUETTE{7,1.7}{2}; \ETIQUETTE{7,-0.1}{1};
\ETIQUETTE{8,1.7}{1}; \ETIQUETTE{8,-0.1}{4};
\ETIQUETTE{9,1.7}{5}; \ETIQUETTE{9,-0.1}{1};
\ETIQUETTE{10,1.7}{1}; \ETIQUETTE{10,-0.1}{3};
\ETIQUETTE{11,1.7}{2}; \ETIQUETTE{11,-0.1}{1};
\ETIQUETTE{12,1.75}{1}; \ETIQUETTE{12,-0.1}{4};

\ETIQUETTE{0.5,6.75}{3}; \ETIQUETTE{0.5,4.9}{1};
\ETIQUETTE{1.5,6.7}{1}; \ETIQUETTE{1.5,4.9}{5};
\ETIQUETTE{2.5,6.7}{2}; \ETIQUETTE{2.5,4.9}{1};
\ETIQUETTE{3.5,6.7}{1}; \ETIQUETTE{3.5,4.9}{4};

\ETIQUETTE{6.5,6.7}{3}; \ETIQUETTE{6.5,4.9}{1};
\ETIQUETTE{7.5,6.7}{1}; \ETIQUETTE{7.5,4.9}{5};
\ETIQUETTE{8.5,6.7}{2}; \ETIQUETTE{8.5,4.9}{1};
\ETIQUETTE{9.5,6.7}{1}; \ETIQUETTE{9.5,4.9}{3};
\ETIQUETTE{10.5,6.7}{4}; \ETIQUETTE{10.5,4.9}{1};
\ETIQUETTE{11.5,6.75}{1}; \ETIQUETTE{11.5,4.9}{2};

\ETIQUETTE{1,11.75}{3}; \ETIQUETTE{1,9.9}{1};
\ETIQUETTE{2,11.7}{1}; \ETIQUETTE{2,9.9}{5};
\ETIQUETTE{3,11.7}{2}; \ETIQUETTE{3,9.9}{1};
\ETIQUETTE{4,11.7}{1}; \ETIQUETTE{4,9.9}{4};

\ETIQUETTE{7,11.7}{3}; \ETIQUETTE{7,9.9}{1};
\ETIQUETTE{8,11.7}{2}; \ETIQUETTE{8,9.9}{1};
\ETIQUETTE{9,11.7}{1}; \ETIQUETTE{9,9.9}{5};
\ETIQUETTE{10,11.7}{4}; \ETIQUETTE{10,9.9}{1};
\ETIQUETTE{11,11.75}{1}; \ETIQUETTE{11,9.9}{2};

\ETIQUETTE{1.5,16.75}{3}; \ETIQUETTE{1.5,14.9}{1};
\ETIQUETTE{2.5,16.7}{1}; \ETIQUETTE{2.5,14.9}{5};
\ETIQUETTE{3.5,16.7}{2}; \ETIQUETTE{3.5,14.9}{1};
\ETIQUETTE{4.5,16.7}{1}; \ETIQUETTE{4.5,14.9}{4};

\ETIQUETTE{7.5,16.7}{3}; \ETIQUETTE{7.5,14.9}{1};
\ETIQUETTE{8.5,16.7}{1}; \ETIQUETTE{8.5,14.9}{5};
\ETIQUETTE{9.5,16.7}{2}; \ETIQUETTE{9.5,14.9}{1};
\ETIQUETTE{10.5,16.75}{1}; \ETIQUETTE{10.5,14.9}{4};


\ARETEH{0,1}{3.75}; \ARETEH{5.25,1}{6.75};
\ARETEV{0,0};
\ARETEV{1,0};
\ARETEV{2,0};
\ARETEV{3,0};

\ARETEV{6,0};
\ARETEV{7,0};
\ARETEV{8,0};
\ARETEV{9,0};
\ARETEV{10,0};
\ARETEV{11,0};
\ARETEV{12,0};

\ARETEH{0.5,6}{3.75}; \ARETEH{5.75,6}{5.75};
\ARETEV{0.5,5};
\ARETEV{1.5,5};
\ARETEV{2.5,5};
\ARETEV{3.5,5};

\ARETEV{6.5,5};
\ARETEV{7.5,5};
\ARETEV{8.5,5};
\ARETEV{9.5,5};
\ARETEV{10.5,5};
\ARETEV{11.5,5};

\ARETEH{1,11}{3.75}; \ARETEH{6.25,11}{4.75};
\ARETEV{1,10};
\ARETEV{2,10};
\ARETEV{3,10};
\ARETEV{4,10};

\ARETEV{7,10};
\ARETEV{8,10};
\ARETEV{9,10};
\ARETEV{10,10};
\ARETEV{11,10};

\ARETEH{1.5,16}{3.75}; \ARETEH{6.75,16}{3.75};
\ARETEV{1.5,15};
\ARETEV{2.5,15};
\ARETEV{3.5,15};
\ARETEV{4.5,15};

\ARETEV{7.5,15};
\ARETEV{8.5,15};
\ARETEV{9.5,15};
\ARETEV{10.5,15};


\draw[very thick] (0,1) .. controls (1,2.7) and (11,2.7) .. (12,1);
\draw[very thick] (0.5,6) .. controls (1.5,7.7) and (10.5,7.7) .. (11.5,6);
\draw[very thick] (1,11) .. controls (2,12.7) and (10,12.7) .. (11,11);
\draw[very thick] (1.5,16) .. controls (2.5,17.7) and (9.5,17.7) .. (10.5,16);


\draw[very thick,dashed] (3.5,-0.5) -- (3.5,1.5); \draw[very thick,dashed] (5.5,-0.5) -- (5.5,1.5); 
\ETIQUETTE{4.5,0.15}{...}; \ETIQUETTE{4.5,1.15}{...};

\draw[very thick,dashed] (4,4.5) -- (4,6.5); \draw[very thick,dashed] (6,4.5) -- (6,6.5); 
\ETIQUETTE{5,5.15}{...}; \ETIQUETTE{5,6.15}{...};

\draw[very thick,dashed] (4.5,9.5) -- (4.5,11.5); \draw[very thick,dashed] (6.5,9.5) -- (6.5,11.5); 
\ETIQUETTE{5.5,10.15}{...}; \ETIQUETTE{5.5,11.15}{...};

\draw[very thick,dashed] (5,14.5) -- (5,16.5); \draw[very thick,dashed] (7,14.5) -- (7,16.5); 
\ETIQUETTE{6,15.15}{...}; \ETIQUETTE{6,16.15}{...};

\ETIQUETTE{6,-0.81}{$C_{n}\odot K_1$, $n\ge 8$, $n\equiv 3\pmod{4}$};
\ETIQUETTE{6,4.19}{$C_{n}\odot K_1$, $n\ge 8$, $n\equiv 2\pmod{4}$};
\ETIQUETTE{6,9.19}{$C_{n}\odot K_1$, $n\ge 8$, $n\equiv 1\pmod{4}$};
\ETIQUETTE{6,14.19}{$C_{n}\odot K_1$, $n\ge 8$, $n\equiv 0\pmod{4}$};

\end{tikzpicture}
\caption{Optimal packing colorings of $C_n\odot K_1$, $n\ge 8$}
\label{fig:CrCn}
\end{center}
\end{figure}

\section{\ESOK{Generalized coronae of undirected paths and cycles}}
\label{sec:generalized-coronae}

As observed in the introduction, we know, by Theorem~\ref{th:sloper}, that
for every integer $p\ge 1$,  
$\pcn(P_n\odot pK_1)\le 6$ if $n\le 34$ and $\pcn(P_n\odot pK_1)\le 7$ otherwise.

When considering generalized coronae of paths or cycles, the following proposition is useful:

\begin{proposition}
\ESOK{Let $P_n=x_1\ldots x_n$, \ESOK{$n\ge 2$}, be a path and $P_n\odot pK_1$, $p\ge 1$, be a generalized corona of $P_n$.
Any packing coloring $\pi$ of $P_n\odot pK_1$ with $\pi(x_i)=1$ for some vertex $x_i$ must use \ESOK{at least $p+3$}
colors if $2\le i\le n-1$, or \ESOK{at least $p+2$} colors if $i\in\{1,n\}$.}

Similarly, if $C_n\odot pK_1$, $p\ge 3$, is a generalized corona of $C_n=y_1\ldots y_n$, then 
any packing coloring $\pi'$ of $C_n\odot pK_1$ with $\pi'(y_i)=1$ for some vertex $y_i$ must use \ESOK{at least $p+3$}
colors.
\label{prop:Pn}
\end{proposition}

\begin{proof}
To see that, simply note that if $\pi(x_i)=1$ then no two neighbors of $x_i$ can receive the same color.
Since the degree of $x_i$ is $p+2$ if $2\le i\le n-1$,
or  $p+1$  if $i\in\{1,n\}$, the claim follows.
The proof if similar for $C_n\odot pK_1$.
\end{proof}

%

In order to describe packing colorings of generalized coronae of paths and cycles, we
will use the following notation in the rest of this paper. Observe first that whenever a vertex
of the path, or the cycle, in any such graph is colored with a color distinct from 1, all
the pendant vertices attached to this vertex can be colored 1. Hence, it is necessary
to give the colors of the pendant vertices only when the color of their neighbor is 1.
In that case, these colors will be given within parenthesis, following the color 1.
Such a sequence of colors, called a {\em pattern}, can thus unambigously describe a packing coloring of a (generalized)
corona of a given path. For instance, the colorings of $P_4\odot K_1$
and $P_5\odot K_1$ given in the previous section (see Figure~\ref{fig:CrPnInf9})
will be denoted by 21(3)41(2) and 21(3)41(3)2, respectively.
For packing colorings of (generalized) coronae of cycles, we will put the whole sequence of
colors in brackets in order to emphasize the fact that the pattern is circular.
For instance, the colorings of $C_5\odot K_1$
and $C_6\odot K_1$ given in the previous section (see Figure~\ref{fig:CrC5C6C7})
will be denoted by [321(5)41(2)] and [31(5)21(3)41(2)], respectively.

Let $u$ and $v$ be two words on the alphabet
of colors, such that $[u]$ is a circular pattern. We will say that the pattern $v$ is {\em compatible with} $[u]$ if
$[uv]$ is a circular pattern.

\medskip

The value of the packing chromatic number of generalized coronae of paths $P_n\odot pK_1$ with $p\ge 4$ is
given by the following theorem:

\begin{theorem}
Let $P_n\odot pK_1$, $p\ge 4$, be a generalized corona of the path $P_n$. Then we have:
$$\ESOK{\pcn(P_n\odot pK_1)= \left\{
			    \begin{array}{ll}
			    	 $2$ & \hbox{if $n=1$,} \\
			    	 $3$ & \hbox{if $n=2$,} \\
			    	 $4$ & \hbox{if $n\in\{3,4\}$,} \\
			    	 $5$ & \hbox{if $5\le n\le 8$,} \\
			    	 $6$ & \hbox{if $9\le n\le 34$,} \\
    			     $7$ & \hbox{otherwise.}
   			 	\end{array}
 			 			\right.
}$$ 
\label{th:PnpK1}
\end{theorem}

\begin{proof}
If $n\le 8$, optimal packing colorings of $P_n\odot pK_1$ are given by the patterns
2, 23, 234, 2342, 23425, 234253, 2342532 and 23425324, respectively.

Note that 23425324 is the longest pattern on five colors which do not use color~1 and, moreover,
none of the patterns 123425324 or 234253241 can be used for coloring $P_9\odot 4K_1$ (the pendant
neighbors of vertices with color~1 cannot be colored). Therefore,
$\pcn(P_9\odot pK_1)\ge 6$. 
In~\cite{S04}, Sloper exhibited the following pattern of length~34, which uses colors~2 to~6, and proved
that no such pattern of greater length exists:
$$23425\ 62342\ 53264\ 23524\ 62352\ 43265\ 2342.$$
As before, this pattern cannot be extended by adding color~1 to the left or to the right, so that $\pcn(P_{35}\odot pK_1)\ge 7$. 
Sloper also gave the circular pattern 
$$[23425\ 62342\ 57],$$
 of length~12, that uses colors~2 to~7, which can be used when $n\ge 35$.
By Proposition~\ref{prop:Pn}, all these colorings are optimal.
\end{proof}

\ESOK{The value of the packing chromatic number of generalized coronae of paths $P_n\odot pK_1$, when $p\in\{2,3\}$, 
is given by the next two results. We will see that the maximum value of the packing 
chromatic number of such graphs is 6, slightly better than the bound given in Theorem~\ref{th:PnpK1}.
This is due to the fact that the number of pendant vertices is now bounded by 3, which allows us
to use color 1 for coloring the vertices of the path $P_n$.}

\begin{theorem}
Let $P_n\odot 2K_1$ be a generalized corona of the path $P_n$. Then we have:
$$\ESOK{\pcn(P_n\odot 2K_1)= \left\{
			    \begin{array}{ll}
			    	 $2$ & \hbox{if $n=1$,} \\
			    	 $3$ & \hbox{if $n=2$,} \\
			    	 $4$ & \hbox{if $n\in\{3,4\}$,} \\
			    	 $5$ & \hbox{if $5\le n\le 11$,} \\
    			     $6$ & \hbox{otherwise.}
   			 	\end{array}
 			 			\right.
}$$ 
\label{th:Pn2K1}
\end{theorem}

\begin{proof}
To see that $\pcn(P_n\odot 2K_1)\le 6$ for every $n$, it is enough to use the following
circular pattern of length~12:
$$[1(36)2432\ 56234\ 25].$$
Since $P_m\odot pK_1$ is a subgraph of $P_n\odot pK_1$ for all $m\le n$, 
every packing $\ell$-coloring of $P_n\odot pK_1$ induces a packing $\ell$-coloring
of $P_m\odot pK_1$. Therefore, it suffices to construct optimal packing colorings
of $P_1\odot 2K_1$, $P_2\odot 2K_1$, $P_4\odot 2K_1$ and $P_{11}\odot 2K_1$,
to get that all the claimed values 
are upper bounds.
This can be done by using the  patterns
$2$, $23$, $2342$ and $1(35)243251(23)4231(25)$, 
respectively.

\ESOK{To finish the proof, we need to show that all these bounds are tight.
This is obvious for $n=1$ and this is a direct consequence of Proposition~\ref{prop:Pn}, for $2\le n\le 4$, since
it implies that we cannot use color~1 on the vertices of the path, so that no packing coloring using less
colors than stated in the theorem can exist in those cases.
For $n=5$, Proposition~\ref{prop:Pn} again implies that we cannot use color~1 for the vertices of $P_5$
in a packing 4-coloring and it is easily checked that no such pattern exists (the longest one is $2342$).
Finally, we have to check that there exists no packing 5-coloring of $P_{12}\odot 2K_1$. We did it by means of
a computer program.
}%
\end{proof}

\begin{theorem}
Let $P_n\odot 3K_1$ be a generalized corona of the path $P_n$. Then we have:
$$\ESOK{\pcn(P_n\odot 3K_1)= \left\{
			    \begin{array}{ll}
			    	 $2$ & \hbox{if $n=1$,} \\
			    	 $3$ & \hbox{if $n=2$,} \\
			    	 $4$ & \hbox{if $n\in\{3,4\}$,} \\
			    	 $5$ & \hbox{if $5\le n\le 8$,} \\
    			     $6$ & \hbox{otherwise.}
   			 	\end{array}
 			 			\right.
}$$ 
\label{th:Pn3K1}
\end{theorem}

\begin{proof}
To see that $\pcn(P_n\odot 3K_1)\le 6$ for every $n$, it is enough to consider the following
circular pattern of length~14:
$$[1(234)5234\ 26325\ 4326].$$
\ESOK{As before, it suffices to construct optimal packing colorings
of $P_1\odot 3K_1$, $P_2\odot 3K_1$, $P_4\odot 3K_1$ and $P_8\odot 3K_1$,
to get that all the claimed values 
are upper bounds.
This can be done by using the  patterns 
$2$, $23$, $2342$ and $23425324$,
respectively.
}%

\ESOK{To finish the proof, we need to show that all these bounds are tight.
This is obvious for $n=1$ and this is a direct consequence of Proposition~\ref{prop:Pn}, for $n\in\{2,3,5,9\}$, since
it implies that we cannot use color~1 on the vertices of the path. It is then not difficult
to check that the longest such patterns are the ones given above, and the result follows.
}%
\end{proof}

We now turn to generalized coronae of cycles $C_n\odot pK_1$. 
When $p\ge 4$, we have the following (note the particular case when $n=11$):

\begin{theorem}
Let $C_n\odot pK_1$, $p\ge 4$, be a generalized corona of the cycle $C_n$. Then we have:
$$\pcn(C_n\odot pK_1)= \left\{
			    \begin{array}{ll}
			    	 $4$ & \hbox{if $n=3$,} \\
			    	 $5$ & \hbox{if $n=4$,} \\
			    	 $6$ & \hbox{if $n\in\{5,6\}$,} \\
			    	 $8$ & \hbox{if $n=11$,} \\
			    	 $7$ & \hbox{otherwise.} \\
   			 	\end{array}
 			 			\right.
$$ 
\label{th:Cn4K1}
\end{theorem}

\begin{proof}
Note first that by Proposition~\ref{prop:Pn}, since $p\ge 4$, color 1 cannot be used 
on the vertices of $C_n$ in any packing coloring of $C_n\odot pK_1$ using at most 6 colors.

Packing colorings of $C_n\odot pK_1$, for $3\le n\le 6$, are given by the following circular patterns:
$$[234]\ \ \ [2345]\ \ \ [23456]\ \ \ [234256].$$
It is not difficult to check that these packing colorings are optimal.

On the other hand, a packing 8-coloring of $C_{11}\odot pK_1$ is given by the following circular pattern:
$$[23425324678].$$
Let us show that no packing 7-coloring of $C_{11}\odot pK_1$ can exist.
If color 1 is not used then, due to the length of the cycle, color 2 can be used
at most three times, colors 3 and 4 at most twice each, and colors 5, 6 and 7 at most once each.
Hence, at most 10 vertices of the cycle can be colored.
Now, if color 1 is used on the cycle, then the pendant vertices must be colored 2, 3, 4 and 5, as otherwise
the packing coloring cannot be extended far enough. The coloring is then ``forced'' around the color 1
as $\dots 43271(2345)6234\dots$. It is then easy to check that this pattern cannot be
extended to a packing 7-coloring of $C_{11}\odot pK_1$ (the smallest extension has length
14 and is given by $[43271(2345)623425362]$).

Packing 7-colorings of $C_n\odot pK_1$, for $7\le n\le 15$, $n\neq 11$, are given by the following
circular patterns:
$$\begin{array}{rl}
n=7: & [2342567]; \\
n=8: & [23425367]; \\
n=9: & [234253267]; \\
n=10: & [2342532467]; \\
n=12: & [234253246257]; \\
n=13: & [2342532462357]; \\
n=14: & [23425362432576]; \\
n=15: & [234253264235276]. 
\end{array}$$

Moreover, all the above circular patterns for $n\ge 9$ are compatible with the circular
pattern $[23425367]$ of length 8. Hence, if $n\ge 16$, $n=8q+r$ with $0\le r\le 7$, $r\neq 3$,
a packing 7-coloring of $C_n\odot pK_1$ can be obtained by combining $q-1$ patterns of length 8
followed by a pattern of length $q+r$ (if $r=0$, we thus have $q$ occurrences of the pattern of length 8).

Finally, for $n=8q+3$, $q\ge 2$, a packing 7-coloring of $C_n\odot pK_1$ can be obtained by combining $q-2$ patterns of length 8
followed by the circular pattern $[2342532462352432657]$ of length 19, which is also compatible with $[23425367]$.
This concludes the proof.
\end{proof}

We now consider the remaining cases, that is $p\in\{2,3\}$.
For $p=2$, we have the following (note the particular case when $n=9$):

\begin{theorem}
Let $C_n\odot 2K_1$ be a generalized corona of the cycle $C_n$. Then we have:
$$\pcn(C_n\odot 2K_1)= \left\{
			    \begin{array}{ll}
			    	 $4$ & \hbox{if $n=3$,} \\
			    	 $5$ & \hbox{if $n=4$,} \\
			    	 $7$ & \hbox{if $n=9$,} \\
			    	 $6$ & \hbox{otherwise.} \\
   			 	\end{array}
 			 			\right.
$$ 
\label{th:Cn2K1}
\end{theorem}

\begin{proof}
The packing colorings of $C_n\odot 2K_1$, for $n\le 13$, $n\neq 9$ are given by the following
circular patterns:
$$\begin{array}{rl}
n=3: & [234]; \\
n=4: & [2345]; \\
n=5: & [23456]; \\
n=6: & [234256]; \\
n=7: & [1(23)423526]; \\
n=8: & [1(24)3251(24)326]; \\
n=10: & [1(23)41(23)523421(35)6]; \\
n=11: & [1(23)4231(25)624325]; \\
n=12: & [1(23)41(23)521(26)423526]; \\
n=13: & [1(23)41(23)5231(26)423526]. \\
\end{array}$$
It is not difficult to check that these colorings are optimal for $n\le 6$.
For $n\ge 7$, any packing 5-coloring of $C_n\odot 2K_1$ would induce a packing
5-coloring of $P_{12}\odot 2K_1$, in contradiction with Theorem~\ref{th:Pn2K1}. 

We now consider the case $n\ge 14$. Similarly, no packing 5-coloring of
$C_n\odot 2K_1$ can exist in this case. 
All the patterns given above for $n\ge 8$ are compatible with the circular pattern  $[1(23)423526]$ of length 7.
Moreover, the pattern $423524326$ of length 9 is also compatible with the same pattern $[1(23)423526]$.
This allows us to construct a packing 6-coloring of any generalized corona
$C_n\odot 2K_1$ with $n\ge 14$. If $n=7q+r$, with $q\ge 2$ and $0\le r<7$, the coloring
is obtained by repeating $q-1$ times the pattern $u$ of length 7 and adding the compatible pattern
of length $7+r$ (note that since the pattern $u$ is a circular pattern, it is compatible with itself).

The last case to consider is the case $n=9$.
A packing 7-coloring of $C_9\odot 2K_1$ is given by the circular pattern 
$$[1(24)3251(24)3267].$$
It is then tedious but not difficult to check that $C_9\odot 2K_1$
does not admit any packing 6-coloring.
 (The main idea is that in such a case,
each of the colors 4, 5 and 6 can be used only once
on the vertices of $C_9$
while the color 3 can be used at most twice and the color 2 at most three times,
so that color 1 has to be used on some vertex of $C_9$;
but in that case, the colors assigned to the pendant neighbors of this vertex
forces the color 1 to be used again on the cycle, leading eventually to a contradiction.)
\end{proof}

Finally, for $p=3$, we have the following:
\begin{theorem}
Let $C_n\odot 3K_1$ be a generalized corona of the cycle $C_n$. Then we have:
$$\pcn(C_n\odot 3K_1)= \left\{
			    \begin{array}{ll}
			    	 $4$ & \hbox{if $n=3$,} \\
			    	 $5$ & \hbox{if $n=4$,} \\
			    	 $7$ & \hbox{if $n\in\{7,\dots,13,15,\dots,22,24,\dots,27,30,\dots,36,39,40,41\}$} \\
			    	 	 & \hbox{$\ \ \ \ \cup\ \{45,47,\dots,50,53,54,55,59,62,63,64,68,77,78,91\}$,} \\
			    	 $6$ & \hbox{otherwise.} \\
   			 	\end{array}
 			 			\right.
$$ 
\label{th:Cn3K1}
\end{theorem}

\begin{proof}
By Theorem~\ref{th:Cn4K1} and Proposition~\ref{prop:subgraph}, we know that $\pcn(C_n\odot 3K_1)\le 7$ for every $n\ge 3$, $n\neq 11$.
Packing colorings of $C_3\odot 3K_1$, $C_4\odot 3K_1$, $C_5\odot 3K_1$ and $C_6\odot 3K_1$ are given by the following circular patterns:
$$[234],\ \ \ [2345],\ \ \ [23456],\ \ \ [234256],$$
whose optimality is easy to check.

\ESOK{Table~\ref{tab:Cn3K1} gives, as circular patterns, packing 6-colorings of $C_n\odot 3K_1$
for every $n\in\{14,23,29,38,44,46,61,67,69,73,76,82,92\}$ (pendant neighbors of vertices
colored~1 are always assigned colors 2, 3 and 4).
Since all these patterns begin with $152342\dots$ and end with $\dots 524326$,
they are all pairwise compatible.
Therefore, by repeating the pattern of length 14 a certain number of times, and adding one of
the patterns of Table~\ref{tab:Cn3K1}, we can produce a packing 6-coloring
of $C_n\odot 3K_1$ in all the following cases, according to the value of $n$ mod 14:
\begin{itemize}
\item $n=14q$, $n\ge 14$,
\item $n=14q+1$, $n\ge 29$ (by repeating $q-2$ times the pattern of length 14 and adding the pattern of length 29),
\item $n=14q+2$, $n\ge 44$ (by repeating $q-3$ times the pattern of length 14 and adding the pattern of length 44),
\item $n=14q+3$, $n\ge 73$ (by repeating $q-5$ times the pattern of length 14 and adding the pattern of length 73),
\item $n=14q+4$, $n\ge 46$ (by repeating $q-3$ times the pattern of length 14 and adding the pattern of length 46),
\item $n=14q+5$, $n\ge 61$ (by repeating $q-4$ times the pattern of length 14 and adding the pattern of length 61),
\item $n=14q+6$, $n\ge 76$ (by repeating $q-5$ times the pattern of length 14 and adding the pattern of length 76),
\item $n=14q+7$, $n\ge 105$ (by repeating $q-7$ times the pattern of length 14 and adding the patterns of length 44 and 61),
\item $n=14q+8$, $n\ge 92$ (by repeating $q-6$ times the pattern of length 14 and adding the pattern of length 92),
\item $n=14q+9$, $n\ge 23$ (by repeating $q-1$ times the pattern of length 14 and adding the pattern of length 23),
\item $n=14q+10$, $n\ge 38$ (by repeating $q-2$ times the pattern of length 14 and adding the pattern of length 38),
\item $n=14q+11$, $n\ge 67$ (by repeating $q-4$ times the pattern of length 14 and adding the pattern of length 67),
\item $n=14q+12$, $n\ge 82$ (by repeating $q-5$ times the pattern of length 14 and adding the pattern of length 82),
\item $n=14q+13$, $n\ge 69$ (by repeating $q-4$ times the pattern of length 14 and adding the pattern of length 69).
\end{itemize}
It is now easy to check that the remaining values of $n$, for which a packing 6-coloring cannot be produced in this way, are exactly
those given in the statement of the theorem. The fact that, for each of these values, $\pcn(C_n\odot 3K_1)=7$
has been checked by means of a computer program.}
\end{proof}

\begin{table}[t]
\small
\begin{center}
  \begin{tabular}{cl}
\hline
{\bf n} & {\bf circular pattern} \\
\hline
\hline
14 & [1523426325 4326] \\
\hline
23 & [1523426324 5236423524 326] \\
\hline
29 & [1523426324 5236423524 623524326] \\
\hline
38 & [1523426324 5236243251 6234253246 23524326] \\
\hline
44 & [1523426324 5236243251 6234253264 2352462352 4326] \\
\hline
46 & [1523426324 5236423524 3261523426 3245236423 524326] \\
\hline
61 & [1523426324 5236243251 6234253246 2352432615 2342632452 3642352432 6] \\
\hline
67 & [1523426324 5236243251 6234253246 2352432615 2342632452 3642352462 \\
 & 3524326] \\
\hline
69 & [1523426324 5236423524 3261523426 3245236423 5243261523 4263245236 \\ 
 & 423524326] \\
\hline
73 & [1523426324 5236243251 6234253264 2352462352 4326152342 6324523642 \\
 & 3524623524 326] \\
\hline
76 & [1523426324 5236243251 6234253246 2352432615 2342632452 3624325162 \\
 & 3425324623 524326] \\
\hline
82 & [1523426324 5236243251 6234253246 2352432615 2342632452 3624325162 \\
 & 3425326423 5246235243 26] \\
\hline
92 & [1523426324 5236423524 3261523426 3245236423 5243261523 4263245236 \\
 & 4235243261 5234263245 2364235243 26] \\
\hline
  \end{tabular}
\end{center}
\caption{Circular patterns for the proof of Theorem~\ref{th:Cn3K1}}
\label{tab:Cn3K1}
\end{table}

\section{\ESOK{Oriented paths, oriented cycles and their generalized coronae}}
\label{sec:oriented-coronae}

In this section, we extend  the notion of packing colorings to digraphs and study the case of oriented
graphs whose underlying undirected graph is a path, a cycle, or a generalized corona of a path or a cycle.

Let $\OR{D}$ be a digraph, with vertex set $V(\OR{D})$ and arc set $E(\OR{D})$.
A {\em directed path} of length $k$ in $\OR{D}$ is a sequence $u_0\dots u_k$ of vertices of $V(\OR{D})$ such that for every $i$,
$0\le i\le k-1$, $u_iu_{i+1}$ is an arc in $E(\OR{D})$.
The {\em weak directed distance} between two vertices $u$ and $v$ in $\OR{D}$, denoted $d_{\OR{D}}(u,v)$,
is the shortest length (number of arcs) of a directed path in $\OR{D}$ going either from $u$ to $v$ or from $v$ to $u$.

A {\em packing $k$-coloring} of a digraph $\OR{D}$ is a mapping $\pi:V(\OR{D})\rightarrow\{1,\ldots,k\}$
such that, for every two distinct vertices $u$ and $v$, $\pi(u)=\pi(v)=i$ implies $d_{\OR{D}}(u,v)>i$.
The {\em packing chromatic number} $\pcn(\OR{D})$ of $\OR{D}$ is then the smallest $k$ such that
$\OR{D}$ admits a packing $k$-coloring.

A digraph $\OR{O}$ with no pair of opposite arcs, that is $uv\in E(\OR{O})$ implies $vu\not\in E(\OR{O})$, is
called an {\em oriented graph}.
If $G$ is an undirected graph, an {\em orientation} of $G$ is any oriented graph $\OR{G}$ obtained by giving
to each edge of $G$ one of its two possible orientations.

By definition, if $\OR{G}$ is any orientation of an undirected graph $G$ then, for any two vertices $u$ and $v$
in $G$, $d_{\OR{G}}(u,v)\le d_G(u,v)$. Therefore, every packing coloring of $G$ is a packing coloring of $\OR{G}$.
Hence, we have the following:

\begin{proposition}
For every orientation $\OR{G}$ of an undirected graph $G$, $\pcn(\OR{G})\le\pcn(G)$.
\label{prop:oriented}
\end{proposition}

Note also that Proposition~\ref{prop:subgraph} is still valid for oriented graphs:

\begin{proposition}
If $\OR{H}$ is a subgraph of $\OR{G}$, then $\pcn(\OR{H})\le\pcn(\OR{G})$.
\label{prop:or-subgraph}
\end{proposition}

\ESOK{The characterization of oriented graphs with packing chromatic number 2 is given by the following
result:}

\begin{proposition}
\ESOK{For every orientation $\OR{G}$ of an undirected graph $G$, $\pcn(\OR{G})=2$
if and only if (i) $G$ is bipartite and (ii) one part of the bipartition of $G$
contains only sources or sinks in $\OR{G}$.}
\label{prop:pcn2}
\end{proposition}

\begin{proof}
\ESOK{Clearly, $\pcn(\OR{G})>2$ whenever $G$ is not bipartite.
Assume thus that $G$ is bipartite. Since color 1 cannot be used for the central vertex of any directed path of length 2, 
we get that $\pcn(\OR{G})=2$  if and only if all the vertices from one of the two parts
are sources or sinks in $\OR{G}$.}
\end{proof}

We now determine the packing chromatic number of orientations of paths, cycles, and coronae of paths and cycles.

For oriented paths, we have the following:

\begin{theorem}
Let $\OR{P_n}$ be any orientation of the path $P_n=x_1\dots x_n$. Then, for every $n\ge 2$,
$2\le \pcn(\OR{P_n}) \le 3$.
Moreover, $\pcn(\OR{P_n})=2$ if and only if 
\ESOK{one part of the bipartition of $P_n$
contains only sources or sinks in $\OR{P_n}$}.
\label{th:or-Pn}
\end{theorem}

\begin{proof}
Since adjacent vertices cannot receive the same color, we clearly have $\pcn(\OR{P_n})\ge 2$ for all $n\ge 2$.
\ESOK{By Theorem~\ref{th:goddard}, we know that $\pcn(P_n)\le 3$ for every $n\ge 2$ and thus, by Proposition~\ref{prop:oriented}, we get
that $\pcn(\OR{P_n})\le 3$ for every $n\ge 2$.}

\ESOK{The last claim directly follows from Proposition~\ref{prop:pcn2}.}
\end{proof}
%

%

For oriented cycles, we have the following:

\begin{theorem}
Let $\OR{C_n}$ be any orientation of the cycle $C_n=x_0\dots x_{n-1}x_0$. Then, for every $n\ge 3$,
$2\le \pcn(\OR{C_n}) \le 4$.
Moreover, 
\begin{itemize}
\item[{\rm (1)}]  $\pcn(\OR{C_n})=2$ if and only if 
\ESOK{$C_n$ is bipartite (that is, $n$ is even) and one part of the bipartition 
contains only sources or sinks in $\OR{C_n}$}.
%
\item[{\rm (2)}] $\pcn(\OR{C_n})=4$ if and only if $\OR{C_n}$ is a directed cycle (all arcs have the same direction), 
   $n\ge 5$ and $n\not\equiv 0\pmod 4$.
   \end{itemize}
\label{th:or-Cn}
\end{theorem}

\begin{proof}
Since adjacent vertices cannot receive the same color, we clearly have $\pcn(\OR{C_n})\ge 2$ for all $n\ge 3$.
By Theorem~\ref{th:goddard}, we know that $\pcn(C_n)\le 4$ for every $n\ge 3$ and thus, by Proposition~\ref{prop:oriented}, we get
that $\pcn(\OR{C_n})\le 4$ for every $n\ge 3$.

\ESOK{Claim (1) directly follows from Proposition~\ref{prop:pcn2}.}


\ESOK{Let us now consider Claim (2).}
By Theorem~\ref{th:goddard}, we know that $\pcn(C_n)=4$ if and only if $n\ge 5$ and $n\not\equiv 0\pmod 4$.
\ESOK{By Proposition~\ref{prop:oriented}, we get that $\pcn(\OR{C_n})\le 3$ in all other cases.
Thus suppose that $n\ge 5$ and $n\not\equiv 0\pmod 4$.}
If $\OR{C_n}$ is a directed cycle, with all arcs having the same direction, then $d_{\OR{C_n}}(x_i,x_j)=d_{C_n}(x_ix_j)$
for every $0\le i,j\le n-1$ and thus $\pcn(\OR{C_n})=4$.
If $\OR{C_n}$ is not a directed cycle, it contains a source vertex, say $x_0$ without loss of generality.
\ESOK{We will prove that, in this case, $\OR{C_n}$ admits a packing 3-coloring.}

\ESOK{We consider three cases:} 
\begin{itemize}
\item If $n\equiv 1\pmod 4$, a packing 3-coloring of $\OR{C_n}$ is given by the following pattern:
$$1231\ |\ 2131\ |\ \dots\ |\ 2131\ |\ 2.$$
\item If $n\equiv 2\pmod 4$, a packing 3-coloring of $\OR{C_n}$ is given by the following pattern:
$$1\ |\ 2131\ |\ \dots\ |\ 2131\ |\ 2.$$
\item If $n\equiv 3\pmod 4$, a packing 3-coloring of $\OR{C_n}$ is given by the following pattern:
$$13\ |\ 1213\ |\ \dots\ |\ 1213\ |\ 2.$$
\end{itemize}


This completes the proof.
\end{proof}

For orientations of generalized coronae of paths, we have the following:

\begin{theorem}
Let $\OR{G}$ be any orientation of a generalized corona $P_n\odot pK_1$, with $p\ge 1$ and $P_n=x_1\dots x_n$. Then, \ESOK{for every $n\ge 1$},
$2\le \pcn(\OR{G}) \le 3$.
Moreover, $\pcn(\OR{G})=2$ if and only if \ESOK{one part of the bipartition of $P_n\odot pK_1$ contains only sources or sinks in $\OR{G}$}.
\label{th:or-PnpK1}
\end{theorem}

\begin{figure}
\begin{center}
\begin{tikzpicture}[domain=0:17,x=0.8cm,y=0.8cm]
\begin{scope}[shift={(0,0)}]
\SOMMET{0,0}; \SOMMET{0,1};
\SOMMET{1,0}; \SOMMET{1,1};
\SOMMET{2,0}; \SOMMET{2,1};
\SOMMET{3,0}; \SOMMET{3,1};
\SOMMET{4,0}; \SOMMET{4,1};
\SOMMET{5,0}; \SOMMET{5,1};
\SOMMET{6,0}; \SOMMET{6,1};
\SOMMET{7,0}; \SOMMET{7,1};

\ETIQUETTE{0,1.7}{1}; \ETIQUETTE{0,-0.1}{2};
\ETIQUETTE{1,1.7}{3}; \ETIQUETTE{1,-0.1}{1};
\ETIQUETTE{2,1.7}{1}; \ETIQUETTE{2,-0.1}{3};
\ETIQUETTE{3,1.7}{2}; \ETIQUETTE{3,-0.1}{1};
\ETIQUETTE{4,1.7}{1}; \ETIQUETTE{4,-0.1}{3};
\ETIQUETTE{5,1.7}{3}; \ETIQUETTE{5,-0.1}{1};
\ETIQUETTE{6,1.7}{1}; \ETIQUETTE{6,-0.1}{2};
\ETIQUETTE{7,1.7}{2}; \ETIQUETTE{7,-0.1}{1};

\ARCE{1,1};
\ARCE{2,1};
\ARCE{3,1};
\ARCO{3,1};
\ARCO{4,1};
\ARCE{6,1};
\ARCE{7,1};

\ARCS{0,1};
\ARCS{2,1};
\ARCS{3,1};
\ARCS{4,1};
\ARCN{1,0};
\ARCN{5,0};
\ARCN{6,0};
\ARCN{7,0};

\ARETEH{0,1}{7};
\ARETEV{0,0};
\ARETEV{1,0};
\ARETEV{2,0};
\ARETEV{3,0};
\ARETEV{4,0};
\ARETEV{5,0};
\ARETEV{6,0};
\ARETEV{7,0};

\ETIQUETTE{3.5,-0.7}{(a)};

\end{scope}

\begin{scope}[shift={(9,0)}]
\SOMMET{0,0}; \SOMMET{0,1};
\SOMMET{1,0}; \SOMMET{1,1};
\SOMMET{2,0}; \SOMMET{2,1};
\SOMMET{3,0}; \SOMMET{3,1};
\SOMMET{4,0}; \SOMMET{4,1};
\SOMMET{5,0}; \SOMMET{5,1};
\SOMMET{6,0}; \SOMMET{6,1};
\SOMMET{7,0}; \SOMMET{7,1};

\SOMMET{0,2};
\SOMMET{1,2};
\SOMMET{2,2};
\SOMMET{3,2};
\SOMMET{4,2};
\SOMMET{5,2};
\SOMMET{6,2};
\SOMMET{7,2};

\ETIQUETTE{-0.3,1.7}{1}; \ETIQUETTE{0,-0.1}{2};
\ETIQUETTE{0.7,1.7}{3}; \ETIQUETTE{1,-0.1}{1};
\ETIQUETTE{1.7,1.7}{1}; \ETIQUETTE{2,-0.1}{3};
\ETIQUETTE{2.7,1.7}{2}; \ETIQUETTE{3,-0.1}{1};
\ETIQUETTE{3.7,1.7}{1}; \ETIQUETTE{4,-0.1}{3};
\ETIQUETTE{4.7,1.7}{3}; \ETIQUETTE{5,-0.1}{1};
\ETIQUETTE{5.7,1.7}{1}; \ETIQUETTE{6,-0.1}{2};
\ETIQUETTE{6.7,1.7}{2}; \ETIQUETTE{7,-0.1}{1};

\ETIQUETTE{0,2.75}{3};
\ETIQUETTE{1,2.75}{1};
\ETIQUETTE{2,2.75}{2};
\ETIQUETTE{3,2.75}{1};
\ETIQUETTE{4,2.75}{2};
\ETIQUETTE{5,2.75}{1};
\ETIQUETTE{6,2.75}{3};
\ETIQUETTE{7,2.75}{1};

\ARCE{1,1};
\ARCE{2,1};
\ARCE{3,1};
\ARCO{3,1};
\ARCO{4,1};
\ARCE{6,1};
\ARCE{7,1};

\ARCS{0,1};
\ARCS{2,1};
\ARCS{3,1};
\ARCS{4,1};
\ARCN{1,0};
\ARCN{5,0};
\ARCN{6,0};
\ARCN{7,0};

\ARCN{1,1};
\ARCN{5,1};
\ARCN{6,1};
\ARCN{7,1};
\ARCS{0,2};
\ARCS{2,2};
\ARCS{3,2};
\ARCS{4,2};

\ARETEH{0,1}{7};
\ARETEV{0,0};
\ARETEV{1,0};
\ARETEV{2,0};
\ARETEV{3,0};
\ARETEV{4,0};
\ARETEV{5,0};
\ARETEV{6,0};
\ARETEV{7,0};
\ARETEV{0,1};
\ARETEV{1,1};
\ARETEV{2,1};
\ARETEV{3,1};
\ARETEV{4,1};
\ARETEV{5,1};
\ARETEV{6,1};
\ARETEV{7,1};

\ETIQUETTE{3.5,-0.7}{(b)};

\end{scope}

\end{tikzpicture}
\caption{Packing colorings for the proof of Theorem~\ref{th:or-PnpK1}}
\label{fig:oriented-caterpillar}
\end{center}
\end{figure}
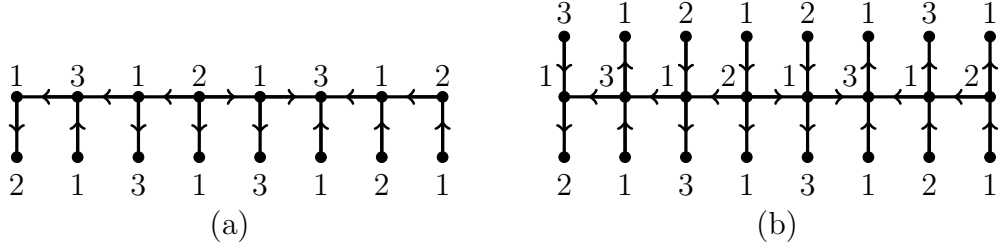

\begin{proof}
Since a packing coloring is a proper coloring, we clearly have 
$\pcn(\OR{G}) \ge 2$ for every orientation $\OR{G}$ of $P_n\odot pK_1$, $n,p\ge 1$.

We first consider the case $p=1$.
For any orientation $\OR{G}$ of $P_1\odot K_1$,  the coloring given by the pattern 1(2),
 is clearly a packing 2-coloring of $\OR{G}$. Assume now that $n\ge 2$ and let $\OR{G}$ be any orientation of $P_n\odot K_1$.
Let $z_1,\dots,z_n$ denote the pendant vertices associated with $x_1,\dots,x_n$, respectively. 
We will construct inductively a packing 3-coloring $\pi$ of $\OR{G}$.
We first set $\pi(x_1):=1$ and $\pi(z_1):=2$.
Assume now that all the vertices $x_1$, $z_1$, $\dots$, $x_i$, $z_i$, $1\le i\le n-1$ have been colored
in such a way that $\pi(x_i)=1$ if and only if $i$ is odd and $\pi(z_i)=1$ if and only if $i$ is even.
Then, use the following rule:
\begin{itemize}
\item If $\pi(x_i)=1$ then set $\pi(x_{i+1}):=5-\pi(z_i)$ if $z_ix_ix_{i+1}$ is a directed path (in either direction)
and $\pi(x_{i+1}):=\pi(z_i)$ otherwise. In both cases, set $\pi(z_{i+1}):=1$.
\item If $\pi(x_i)\neq 1$ then set $\pi(z_{i+1}):=5-\pi(x_i)$ if $x_ix_{i+1}z_{i+1}$ is a directed path (in either direction)
and $\pi(z_{i+1}):=\pi(x_i)$ otherwise. In both cases, set $\pi(x_{i+1}):=1$.
\end{itemize}
The coloring $\pi$ thus obtained (see Figure~\ref{fig:oriented-caterpillar}(a) for an example) has the following property:
\begin{itemize}
\item[(P)] every vertex with
color 1 is such that all its in-neighbors have the same color $\alpha\in\{2,3\}$ and
all its out-neighbors have the same color $5-\alpha\in\{2,3\}$.
\end{itemize}
 The coloring $\pi$ is thus a packing 3-coloring of $\OR{G}$.

Consider now the case $p\ge 2$. We first color the vertices $x_1,\dots,x_n$ and one of their pendant neighbors using
the procedure described above, and then color the remaining pendant vertices in such a way that  property (P)
is satisfied. Hence, all pendant neighbors of a vertex with color 2 or 3 will be colored 1,
and all pendant neighbors of a vertex with color 1 will be colored 2 or 3, depending on the orientation of the
corresponding arc
(see Figure~\ref{fig:oriented-caterpillar}(b) for an example).

\ESOK{The last claim directly follows from Proposition~\ref{prop:pcn2}.}
\end{proof}

Finally, for orientations of generalized coronae of cycles, we have the following:

\begin{theorem}
Let $\OR{G}$ be any orientation of a generalized corona $C_n\odot pK_1$, with $p\ge 1$ and $C_n=x_0\dots x_{n-1}$. Then, for every $n\ge 3$,
$2\le \pcn(\OR{G}) \le 4$.
Moreover, 
\begin{itemize}
\item [{\rm (1)}] $\pcn(\OR{G})=2$ if and only if \ESOK{$C_n\odot pK_1$ is bipartite (that is, $n$ is even) and one part of the bipartition 
contains only sources or sinks in $\OR{G}$}.

\item [{\rm (2)}] $\pcn(\OR{G})=4$ if and only if either:
  \begin{itemize}
  \item[{\rm (2.1)}] $\OR{C_n}$ is a directed cycle, $n\ge 5$ and $n\not\equiv 0\pmod 4$, or
  \item [{\rm (2.2)}] $\OR{G}$ contains the oriented graph depicted in Figure~\ref{fig:subgraphsChiRho4} as a subgraph, or
  \item [{\rm (2.3)}] $n\equiv 0\pmod 4$ and there exists a vertex $x_i$, $0\le i\le n-1$, such that 
  the paths $x_ix_{i+1}x_{i+2}x_{i+3}$ and $x_{i+4}\dots x_{i-1}$(indices are taken modulo~$n$) are both directed paths, but in opposite direction.
  \end{itemize}
\end{itemize}
\label{th:or-CnpK1}
\end{theorem}

\begin{figure}
\begin{center}
\begin{tikzpicture}[domain=0:17,x=0.8cm,y=0.8cm]
\begin{scope}[shift={(0,0)}]
\SOMMET{0,0}; \SOMMET{0,1};
\SOMMET{1,0}; \SOMMET{1,1};
\SOMMET{2,0};
\SOMMET{2,1};

\ARCO{0,1};
\ARCN{0,0};
\ARCE{1,0};
\ARCN{1,0};

\ARCO{1,1};
\ARCE{2,0};

%

\ARETEH{0,0}{2};
\ARETEH{0,1}{2};
\ARETEV{0,0};
\ARETEV{1,0};

\end{scope}

\end{tikzpicture}
\caption{Configuration for the proof of Theorem~\ref{th:or-CnpK1}}
\label{fig:subgraphsChiRho4}
\end{center}
\end{figure}
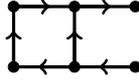

Before proving this theorem, we introduce a useful coloring procedure, called {\em standard coloring procedure} (SCP for short), 
that produces a coloring $\pi$ of an orientation of the path $P_n=x_1\dots x_n$: 

\begin{enumerate}
\item Assume $(c,c')\in\{1,2,3\}^2$, with $|\{c,c'\}\cap\{1\}|=1$, and $S\subseteq V(P_n)$ are given.
\item Set $\pi(x_1):=c$ and $\pi(x_2):=c'$.
\item For $j=3,\dots,n$, set $\pi(x_j):=1$ if $\pi(x_{j-1})\neq 1$, $\pi(x_j):=\pi(x_{j-2})$
if $\pi(x_{j-1})=1$ and $x_{j-1}\in S$, and $\pi(x_j):=5-\pi(x_{j-2})$ otherwise.
\end{enumerate}

Figure~\ref{fig:SCP} shows colorings of two orientations of $P_8=x_1\dots x_8$ produced by SCP, with 
$(c,c')=(1,2)$ and $S=\{x_3\}$, and with $(c,c')=(3,1)$ and $S=\{x_4,x_8\}$, respectively. 
Note that SCP always produces a packing 3-coloring of the path $x_1\dots x_n$, 
but  not necessarily 
a packing 3-coloring of $\OR{C_n}$, and that the only possible conflicts lie on the path 
$x_{n-2}x_{n-1}x_nx_1x_2x_3$
(such conflicts may appear when a directed path of length 2 or 3 contains $x_1$ as an internal vertex).
For instance, the second example depicted in Figure~\ref{fig:SCP} is a 
packing 3-coloring of $\OR{C_8}$, while the first one is not.

Observe that if $c=1$ (resp. $c'=1$) SCP assigns color 1 to every vertex $x_j$ such that $j$ is odd (resp. even),
and colors 2 and 3 alternate on other vertices whenever $S$ is empty.
If $S$ is not empty, we have $|S|$, or $|S|-1$ if $x_1\in S$ and $c=1$ (resp. $x_2\in S$ and $c'=1$), places where the
color 2 or 3 is duplicated. Hence, we have the following:

\begin{figure}
\begin{center}
$$1\RR 2\RR 1\LL 2\LL 1\LL 3\RR 1\RR 2\RR (1)$$
$$3\RR 1\RR 2\RR 1\LL 2\LL 1\LL 3\LL 1\RR (3)$$
\caption{Sample colorings produced by SCP}
\label{fig:SCP}
\end{center}
\end{figure}

\begin{proposition}
Let $\OR{P_n}$ be any orientation of the path $P_n=x_1\dots x_n$ of \ESOK{odd} length~$n-1$
and $S$ be a set of sources or sinks in $\OR{P_n}$ with odd indices not containing $x_1$. Consider the coloring $\pi$ of $\OR{P_n}$
produced by SCP with $(c,c')=(1,\alpha)$ for some $\alpha\in\{2,3\}$ and $S$.
Then we have:
\begin{itemize}
\item[{\rm (i)}] $\pi(x_n)=\alpha$ if $|S|$ is even (resp. odd) and $n\equiv 2\pmod 4$ (resp. $n\equiv 0\pmod 4$),
\item[{\rm (ii)}] $\pi(x_n)=5-\alpha$ otherwise.
\end{itemize}
\label{prop:SCP}
\end{proposition}

\begin{proof}
This directly follows from the above discussion.
\end{proof}

\medskip

\begin{proof}
{\bf [of Theorem~\ref{th:or-CnpK1}]}
Since a packing coloring is a proper coloring, we clearly have 
$\pcn(\OR{G}) \ge 2$ for every orientation $\OR{G}$ of $C_n\odot pK_1$, $n\ge 3$, $p\ge 1$.

Let $\OR{G}$ be any orientation of $C_n\odot pK_1$ and $\OR{C_n}$ be the orientation of the cycle $C_n$
induced by $\OR{G}$. Denote by $z_i^j$, $1\le j\le p$, the pendant neighbors of $x_i$, $0\le i\le n-1$.
We consider two cases.

If $\OR{C_n}$ contains a source vertex, say $x_0$ without loss of generality, then, by Theorem~\ref{th:or-PnpK1}, there exists
a packing 3-coloring of $\OR{G}\setminus\{x_0,z_0^1,\dots,z_0^p\}$. Since $x_0$ is a source, 
this packing coloring can be extended to a packing 4-coloring
of $\OR{G}$ by coloring $x_0$ with color 4 and all vertices $z_0^j$, $1\le j\le p$, with color 1.

If $\OR{C_n}$ does not contain any source vertex then $\OR{C_n}$ is a directed cycle.
By Theorem~\ref{th:or-Cn}, we know that there exists a packing 4-coloring $\pi$ of $\OR{C_n}$.
This packing coloring can be extended to a packing 4-coloring
of $\OR{G}$ by coloring every pendant vertex $z_i^j$, $0\le i\le n-1$, $1\le j\le p$,
by $\pi(x_{i-1})$ if $z_i^jx_i$ is an arc in $\OR{G}$ and by $\pi(x_{i+1})$ otherwise (indices are taken modulo $n$).
Hence, $\pcn(\OR{G}) \le 4$ for every orientation $\OR{G}$ of $C_n\odot pK_1$, $n\ge 3$, $p\ge 1$.

\medskip

\ESOK{Claim (1) directly follows from Proposition~\ref{prop:pcn2}.}



\medskip

We now consider Claim (2).
If $\pcn(\OR{C_n})=4$ (which happens, by Theorem~\ref{th:or-Cn}, if and only if $\OR{C_n}$ is a directed cycle,  $n\ge 5$ and $n\not\equiv 0\pmod 4$)
then,  by Proposition~\ref{prop:or-subgraph}, $\pcn(\OR{G})=4$ (condition 2.1 of the theorem).

If $\pcn(\OR{C_n})=2$ (which happens, by Theorem~\ref{th:or-Cn}, if and only if $n$ is even and the orientation $\OR{C_n}$ of $C_n$
is alternating) then we clearly have $\pcn(\OR{G})\le 3$ since both colors 2 and 3 are available for pendant neighbors of vertices colored 1.

Suppose therefore that $\pcn(\OR{C_n})=3$.
If $\OR{C_n}$ is a directed cycle, which implies $n\equiv 0\pmod 4$, then the packing 3-coloring
given by the circular pattern $[1213]$ can be extended to a packing 3-coloring of $\OR{G}$, as in the proof of Theorem~\ref{th:or-PnpK1}.

Assume now that $\OR{C_n}$ is not a directed cycle and let $\pi$ be a packing 3-coloring of $\OR{C_n}$.
This coloring can be extended to a packing 3-coloring of $\OR{G}$
except if there exists three consecutive vertices $x_{i-1}x_ix_{i+1}$ (indices are taken modulo $n$) such that
(i) $x_i$ is a source (resp. a sink) in $\OR{C_n}$ but not in $\OR{G}$, and (ii) $\pi(x_i)=1$ and $\{\pi(x_{i-1}),\pi(x_{i+1})\}=\{2,3\}$.
Indeed, if such a case occurs, none of the colors from the set $\{1,2,3\}$ can be assigned to a pendant out-neighbor (resp. in-neighbor) of~$x_i$.
Otherwise, the packing 3-coloring of $\OR{C_n}$ can be extended to a packing 3-coloring of $\OR{G}$ by (i) assigning color 1
to all pendant neighbors of vertices colored 2 or 3, (ii) assigning the color $\pi(x_{i-1})$ to every pendant out-neighbor (resp. in-neighbor)
of a source (resp. a sink) vertex $x_i$ of $\OR{G}$ and the color $5-\pi(x_{i-1})$ to its in-neighbors (resp. out-neighbors),
and (iii) assigning the color $\pi(x_{i-1})$ to every pendant in-neighbor (resp. out-neighbor)
of a vertex $x_i$ which is neither a source nor a sink in $\OR{G}$, and the color $\pi(x_{i+1})$ to its out-neighbors (resp. in-neighbors),
whenever $x_{i-1}x_ix_{i+1}$ (resp. $x_{i+1}x_ix_{i-1}$) is a directed path.

We  thus need to determine in which cases the orientation $\OR{C_n}$ of $C_n$ can be colored
in such a way that such a situation does not occur. Such colorings will be called {\em good} packing colorings.

For any subset $X$ of $V(C_n)$, we
denote by $S(X)$ the subset of $X$ containing all the vertices that are either a source or
a sink in $\OR{C_n}$, and by $S^*(X)$ the subset of $S(X)$ containing all the vertices that are neither a source nor
a sink in $\OR{G}$. Hence, $S^*(V(C_n))$ is precisely the set of vertices we must care about.
\ESOK{Obviously, if $S^*(V(C_n))$ is empty, every packing 3-coloring of $\OR{C_n}$ is good.
We thus assume in the rest of the proof that $S^*(V(C_n))$ is not empty.}
Note also that $|S(V(C_n))|$ is even for every orientation $\OR{C_n}$ of $C_n$.

In the following, we will construct good packing 3-colorings, when this is possible, using SCP with an adequate set $S$ either on
the whole cycle $\OR{C_n}$ or on part of it.


\medskip

We consider four cases, according to the value of $n$ mod~4:

\begin{itemize}

\item
{\bf Case~1:} $n\equiv 0\pmod 4$.\\
Consider first the case $n=4$. The only possible packing 3-coloring of any orientation $\OR{C_4}$
of $C_n$ with $\pcn(\OR{C_4})=3$
is 1213. It is then easy to check that the only orientation $\OR{C_4}$ of $C_4$ for which we cannot
produce a good packing 3-coloring is the one given in Figure~\ref{fig:subgraphsChiRho4}.
In the following, we can thus assume $n\ge 8$.

 Since $n$ is even, $C_n$ is bipartite. Let $(A,B)$
denote the bipartition of $V(C_n)$. 
If $|S^*(A)|$ is even or $|S^*(B)|$ is even, a good coloring can be
obtained by means of SCP. Suppose without loss of generality that $A=\{x_0,x_2,\dots,x_{n-2}\}$
and $|S^*(A)|$ is even. Consider the coloring $\pi$ produced by SCP, starting at $x_0$, with $(c,c')=(1,2)$ and $S=S^*(A)$.
Since $n\equiv 0\pmod 4$ and $|S^*(A)|$ is even, by Proposition~\ref{prop:SCP}, $\pi$ is 
 a good packing 3-coloring of $\OR{C_n}$.

If both $|S^*(A)|$ and $|S^*(B)|$ are odd, but $S(A)\setminus S^*(A)\neq\emptyset$
or $S(B)\setminus S^*(B)\neq\emptyset$, we can proceed in a similar way by using, without loss of generality, the set
$S'(A)=S^*(A)\cup\{x_{2j}\}$, for some vertex $x_{2j}\in S(A)\setminus S^*(A)$, instead of the set $S^*(A)$ in SCP
 since $|S'(A)|$ is even.

Finally, suppose that both $|S^*(A)|$ and $|S^*(B)|$ are odd, $S(A)=S^*(A)$
and $S(B)=S^*(B)$, that is, every source or sink in $\OR{C_n}$ is neither a source nor a sink in $\OR{G}$.
We consider two cases:

\begin{itemize}
\item $|S^*(A)|=|S^*(B)|=1$.\\
Without loss of generality, we may assume that $x_0$ is a source and $x_i$, for some odd $i$, $1\le i\le n-1$, is a sink.
Hence, $x_0\dots x_i$ and $x_{n-1}\dots x_i$ are both directed paths of odd length in $\OR{C_n}$.
Suppose first that $i=1$, that is, $x_0$ is a source and $x_1$ is a sink. A good packing 3-coloring of $\OR{C_n}$ is then given by the following pattern
 (the colors of $x_0$ and $x_i=x_1$ are dotted):
$$[\dot{1}\dot{2}\  3121\  \dots\ 3121\  32].$$
Similarly, if $i\equiv 1\pmod 4$, a good packing 3-coloring of $\OR{C_n}$ is then given by:
$$[\dot{1}\ 2131\  \dots\  2131\ \dot{2}\ 3121\  \dots\ 3121\  32].$$
Now, if $i\equiv 3\pmod 4$, $i\ge 7$, a good packing 3-coloring of $\OR{C_n}$ is given by:
$$[\dot{1}23\ 1213\ \dots\ 1213\ \dot{2}\ 1312\ \dots\ 1312].$$

The remaining case is $i=3$, which corresponds to condition (2.3) of the theorem. 
We will prove that in that case $\OR{C_n}$ does not admit any good packing 3-coloring, which implies $\pcn(\OR{C_n})=4$.
Note first that the directed path $\OR{P}=x_0x_{n-1}\dots x_3$ has length $n-3\equiv 1\pmod 4$. Let us consider the possible
 packing 3-colorings of $\OR{P}$. Clearly, the pattern 123 can only be used on the left end of $\OR{P}$, while the pattern
321 can only be used on the right end of $\OR{P}$. Moreover, the only circular good pattern is $[1213]$. Therefore, up to mirror symmetry
(reversing the orientation of every arc of $\OR{C_n}$ gives the same oriented graph),
there are six possible packing 3-colorings of $\OR{P}$, given by the following patterns:
$$1213\ \dots\ 1213\ 12,$$
$$1213\ \dots\ 1213\ 21,$$
$$123\ 1213\ \dots\ 1213\ 121,$$
$$2131\ \dots\ 2131\ 21,$$
$$3121\ \dots\ 3121\ 31,$$
$$3121\ \dots\ 3121\ 32.$$
It is then not difficult to check that none of these colorings can be extended to a good packing 3-coloring of $\OR{C_n}$,
as shown by the following diagrams (the colors of $x_0$ and $x_3$ are dotted):
$$2\ \longleftarrow\ \dot{1}\ \longrightarrow\ ?\ \longrightarrow\ ?\ \longrightarrow\ \dot{2}\ \longleftarrow\ 1$$
$$2\ \longleftarrow\ \dot{1}\ \longrightarrow\ ?\ \longrightarrow\ ?\ \longrightarrow\ \dot{1}\ \longleftarrow\ 2$$
$$2\ \longleftarrow\ \dot{1}\ \longrightarrow\ ?\ \longrightarrow\ ?\ \longrightarrow\ \dot{1}\ \longleftarrow\ 2$$
$$1\ \longleftarrow\ \dot{2}\ \longrightarrow\ ?\ \longrightarrow\ ?\ \longrightarrow\ \dot{1}\ \longleftarrow\ 2$$
$$1\ \longleftarrow\ \dot{3}\ \longrightarrow\ ?\ \longrightarrow\ ?\ \longrightarrow\ \dot{1}\ \longleftarrow\ 3$$
$$1\ \longleftarrow\ \dot{3}\ \longrightarrow\ ?\ \longrightarrow\ ?\ \longrightarrow\ \dot{2}\ \longleftarrow\ 3$$

\item $|S^*(A)|\ge 3$ or $|S^*(B)|\ge 3$.\\
Suppose $|S^*(A)|\ge 3$, without loss of generality.
Since $n\equiv 0\pmod 4$ and both $|S^*(A)|$ and $|S^*(B)|$ are odd, by Proposition~\ref{prop:SCP}, applying SCP starting at $x_0$
leads in all cases to a ``bad'' coloring, that assigns to $x_{n-1}x_0x_1$ either the pattern 213 or 312 if $x_0\in S^*(A)$,
or the pattern 212 or 313 otherwise (in that case, $x_{n-1}x_0x_1$ is a directed path, in either direction).
We thus need to ``correct'' this bad coloring, which can be done by 
replacing a sequence $1\alpha\dots\beta 1$ of the coloring produced by SCP by $1\alpha\dots\beta' 1$ with $\beta'=5-\beta$.

\medskip

We consider three subcases.

\begin{enumerate}

\item {\em There exist $a\in S^*(A)$ and $b\in S^*(B)$ with $d_{\OR{C_n}}(a,b)=1$}.\\
We may assume without loss of generality that $a=x_i$ is a source and $b=x_{i+1}$ is a sink.
Hence, we have the following configuration (--- stands for an arc in either direction):
$$\mbox{---}\ \LL a \RR b \LL \mbox{---}\  \mbox{---}$$
Consider the following coloring of this configuration (the colors of $a$ and $b$ are dotted):
$$1\ \mbox{---}\ 3 \LL \dot{2} \RR \dot{1} \LL 2\ \mbox{---}\  3\ \mbox{---}\ 1$$
If the vertex to the right of $b$ is not a source, 
the remaining part of the cycle is not empty (since $n\equiv 0\pmod 4$) and
this coloring can be extended to a good packing 3-coloring of $\OR{C_n}$ 
by means of SCP. To see that, observe that
SCP would have produced the following bad coloring
on the same configuration (the bad color which implies our claim, since our coloring
has modified this color, appears in bold):
$$1\ \mbox{---}\ 3 \LL \dot{1} \RR \dot{3} \LL 1\ \LL  \mathbf{2}\ \mbox{---}\ 1$$
We finally claim that we can always find some $i$ such that $x_i\in S^*(A)$ (resp. $x_i\in S^*(B)$),
$x_{i+1}\in S^*(B)$ (resp. $x_{i+1}\in S^*(A)$) and $x_{i+2}\notin S^*(A)$ (resp. $x_{i+2}\notin S^*(B)$).
This simply follows from the fact that if no such $i$ exists, then the orientation $\OR{C_n}$ of
$C_n$ is alternating, which implies $\pcn(\OR{C_n})=2$, contrary to our assumption.

\item {\em There exist $a\in S^*(A)$ and $b\in S^*(B)$ with $d_{\OR{C_n}}(a,b)\equiv 1\pmod 4$, $d_{\OR{C_n}}(a,b)\ge 5$, and Subcase~1 does not occur}.\\
Again, we assume without loss of generality that $a=x_i$ is a source and $b=x_j$ is a sink. Since subcase 1 does not
occur, we necessarily have the following configuration:
$$\LL\LL a\RR\RR \dots\ \RR\RR\RR b\LL\LL$$
We then color this configuration as follows (the pattern 2131 is repeated as many times as necessary):
$$1\ \mbox{---}\ 3 \LL 2 \LL \dot{1}\RR (2131)^* \RR \dot{2}\LL 3 \LL 1$$
As in the previous subcase, the remaining part of the cycle is not empty.
Hence, this coloring can be extended to a good packing 3-coloring of $\OR{C_n}$ 
by means of SCP, since SCP would have produced the following bad coloring
on the same configuration:
$$1\ \mbox{---}\ 3 \LL 1 \LL \dot{2}\RR (1312)^* \RR \dot{1}\LL \mathbf{2} \LL 1$$

\item {\em There exist $a\in S^*(A)$ and $b\in S^*(B)$ with $d_{\OR{C_n}}(a,b)\equiv 3\pmod 4$, $d_{\OR{C_n}}(a,b)\ge 7$, and Subcases~1 and~2 do not occur}.\\
This subcase can be solved similarly as the previous one. We have the following configuration:
$$\LL a\RR\RR\RR \dots\ \RR\RR\RR\RR b\LL$$
for which we use the following coloring:
$$1 \LL \dot{2}\RR 3 \RR (1213)^* \RR 2 \RR \dot{1} \LL 2$$
Again, the remaining part of the cycle is not empty and this coloring can be extended to a 
good packing 3-coloring of $\OR{C_n}$ by means of SCP, since
SCP  would have produced the following bad coloring on the configuration:
$$1 \LL \dot{2}\RR 1 \RR (3121)^* \RR 3 \RR \dot{1} \LL \mathbf{3}$$

\item {\em None of the previous cases occurs}.\\
If none of the previous cases occurs, then the vertices of $S^*(A)$ and $S^*(B)$ necessarily
alternate on $\OR{C_n}$ and the weak directed distance between any two consecutive such vertices
equals~3. Hence, $\OR{C_n}$ is a sequence of directed paths of length~3 in opposite directions.
Since $|S^*(A)|=|S(A)|$ is odd, the length of $C_n$ equals $6k$ for some odd~$k$, which contradicts
the assumption $n\equiv 0\pmod 4$. 
Therefore,  this last subcase cannot occur.

\end{enumerate}

\end{itemize}

\item
{\bf Case~2:} $n\equiv 2\pmod 4$.\\
 In this case $C_n$ is again bipartite and, using the same procedure as in Case~1, a
good packing 3-coloring of $\OR{C_n}$ can be produced whenever (i) $|S^*(A)|$ or $|S^*(B)|$ is odd, or
(ii) both $|S^*(A)|$ and $|S^*(B)|$ are even, but $S(A)\setminus S^*(A)\neq\emptyset$
or $S(B)\setminus S^*(B)\neq\emptyset$, where $(A,B)$ denotes the bipartition of $V(C_n)$.

Suppose now that both $|S^*(A)|$ and $|S^*(B)|$ are even (they cannot be both equal to~0), $S(A)=S^*(A)$
and $S(B)=S^*(B)$. In that case, SCP produces a bad coloring of $\OR{C_n}$ and this
coloring can be ``corrected'' in exactly the same way as in Case~1 since, for doing that,
we only need $n$ to be even.

\item
{\bf Case~3:} {\em $n$ is odd}.\\
Consider the set $S=S(V(C_n))$, that is the set of vertices that are either a source or a sink in $\OR{C_n}$.
Without loss of generality, suppose that $x_0$ is a source and consider the coloring $\pi$
produced by SCP on the path $x_0x_1\dots x_{n-1}$, starting at $x_0$, with $(c,c')=(2,1)$ and $S$. If $\pi(x_{n-1})=3$, $\pi$ is a
packing 3-coloring of $\OR{G}$, of the form $21\dots 13$, and we are done.

If $\pi(x_{n-1})=2$ ($\pi$ is not a packing coloring of $\OR{G}$), consider the coloring $\pi'$
produced by SCP on the path $x_1x_2\dots x_{n-1}x_0$, starting at $x_1$, with $(c,c')=(3,1)$ and $S$.
Let now $X$ denote the set of sources or sinks which are assigned color 1 by $\pi$,
and $X'$ the set of sources or sinks which are assigned color 1 by $\pi'$. We clearly have $X\cap X'=\emptyset$
and $X\cup X'=S\setminus\{x_0\}$ (since $x_0$ is a source, $\pi(x_0)\neq 1$ and $\pi'(x_0)\neq 1$).
Therefore, since $|S|$ is even,
we get that $|X[$ and $|X'|$ do not have the same parity. Hence, since $\pi(x_0)=2$ and $\pi(x_{n-1})=2$,
starting with $\pi'(x_1)=3$ necessarily gives $\pi'(x_0)=2$.
This proves that $\pi'$ is a good packing 3-coloring of $\OR{G}$,
of the form $231\dots 1$.

\end{itemize}

This concludes the proof.
\end{proof}

\section{Discussion}

In this paper, we have determined the packing chromatic number of coronae and generalized
coronae of paths and cycles. We also extended to digraphs the notion of packing coloring
and determined the packing chromatic number of orientations of such graphs.

In particular, we have proved that every orientation of a generalized corona of a path
admits a packing 3-coloring.
Using a similar proof, it is not difficult to extend this result to the more general case of
oriented trees (we can inductively construct a packing coloring satisfying the property (P)
such that vertices with color 1 correspond to one part of the bipartition of the tree).
Hence, we also have:

\begin{theorem}
Let $T$ be a tree. For any orientation $\OR{T}$ of $T$, $\pcn(\OR{T})\le 3$.
\label{th:or-tree}
\end{theorem}

Since every caterpillar is a tree, we get that every oriented caterpillar has packing chromatic number
at most 3. However, we leave as an open question the characterization of undirected caterpillars with
packing chromatic number at most~4, 5 and~6 (by Theorem~\ref{th:sloper} we know that every  caterpillar
has packing chromatic number at most 7 and characterizing caterpillars with 
packing chromatic number at most~2 or~3 is easy).

\vskip 1cm

\noindent
{\bf Acknowledgment.} Most of this work has been done while the first author
was visiting LaBRI, thanks to a seven-months PROFAS-B+ grant cofunded by 
the Algerian and French governments, while part of this work has been done while
the second author was visiting LaBRI.



\begin{thebibliography}{99}

\bibitem{A11} G. Argiroffo. 
Polynomial instances of the packing coloring problem. 
{\em Electronic Notes in Discrete Mathematics} 37 (2011), 363--368.

\bibitem{ANT12} G. Argiroffo, G. Nasini and P. Torres. 
The packing coloring problem for $(q,q-4)$-graphs. 
{\em Lecture Notes in Computer Science} 7422 (2012), 309--319.

\bibitem{ANT14} G. Argiroffo, G. Nasini and P. Torres. 
The packing coloring problem for lobsters and partner limited graphs.
{\em Discrete Applied Mathematics} 164 (2014), 373--382.

\bibitem{BKR07} B. Bre\v sar, S. Klav\v zar and D.F. Rall. 
On the packing chromatic number of Cartesian products, hexagonal lattice, and trees. 
{\em Discrete Applied Mathematics} 155 (2007), 2303--2311.

\bibitem{EFHL10} J. Ekstein, J. Fiala, P. Holub and B. Lidick\'y.
The packing chromatic number of the square lattice is at least 12.  
March 12, 2010, arXiv:1003.2291v1 [cs.DM].

\bibitem{EHL12} J. Ekstein, P. Holub and B. Lidick\'y. 
Packing chromatic number of distance graphs.
{\em Discrete Applied Mathematics} 160 (2012) 518--524.

\bibitem{EHT14} J. Ekstein, P. Holub and O. Togni.
The packing coloring of distance graphs $D(k,t)$.
{\em Discrete Applied Mathematics} 167 (2014), 100--106.

\bibitem{FG10} J. Fiala and P. A. Golovach. 
Complexity of the packing coloring problem for trees.
{\em Discrete Applied Mathematics} 158 (2010), 771--778.

\bibitem{FKL09} J. Fiala, S. Klav\v zar and B. Lidick\'y.
The packing chromatic number of infinite product graphs.
{\em European Journal of Combinatorics} 30 (2009), 1101--1113.

\bibitem{FR10} A.S. Finbow and D. F. Rall.
On the packing chromatic number of some lattices.
{\em Discrete Applied Mathematics} 158 (2010), 1224--1228.

\bibitem{G15} N. Gastineau.
Dichotomies properties on computational complexity of S-packing coloring problems.
{\em Discrete Mathematics} 338 (2015), no. 6, 1029--1041.

\bibitem{GHHHR03} W. Goddard, S.M. Hedetniemi, S.T. Hedetniemi, J.M. Harris and D.F. Rall. 
Broadcast chromatic numbers of graphs.
The 16th Cumberland Conference on Combinatorics,
Graph Theory, and Computing (2003).

\bibitem{GHHHR08} W. Goddard, S.M. Hedetniemi, S.T. Hedetniemi, J.M. Harris and D.F. Rall. 
Broadcast chromatic numbers of graphs.
{\em Ars Combinatoria} 86 (2008), 33--49.

\bibitem{KV14} D. Kor\v ze and A. Vesel. 
On the packing chromatic number of square and hexagonal lattice. 
{\em Ars Mathematica Contemporanea} 7 (2014), 13--22.

\bibitem{L10} D. La\"iche.
Sur les nombres broadcast chromatiques. 
Magister thesis, University of Sciences and Technology Houari Boumediene, Algeria, 2010
(in french).

\bibitem{RBFK08} D.F. Rall, B. Bre\v sar, A.S. Finbow and S. Klav\v zar. 
On the packing chromatic number of trees, Cartesian products and some infinite graphs.
{\em Electronic Notes in Discrete Mathematics} 30 (2008), 57--61.

\bibitem{S04} C. Sloper. 
An eccentric coloring of trees.
{\em Australasian Journal of  Combinatorics} 29 (2004), 309--321.

\bibitem{SH10}  R. Soukal and P. Holub. 
A note on the packing chromatic number of the square lattice.
{\em Electronic Journal of  Combinatorics} 17 (2010).

\bibitem{T14} O. Togni. 
On packing colorings of distance graphs. 
{\em Discrete Applied Mathematics} 167 (2014), 280--289.

\bibitem{TVP15} P. Torres and M. Valencia-Pabon. 
The packing chromatic number of hypercubes. 
{\em Discrete Applied Mathematics} 190-191 (2015), 127--140.

\bibitem{WRR14b} A. William, I. Rajasingh and S. Roy.
Packing chromatic number of enhanced hypercubes.
{\em International Journal of Mathematics and its Applications} 2 (2014), no. 3, 1--6. 

\bibitem{WRR14a} A. William, S. Roy and I. Rajasingh.
Packing chromatic number of cycle related graphs.
{\em International Journal of Mathematics and Soft Computing} 4 (2014), no. 1, 27--33. 



\end{thebibliography}
\end{document}